\documentclass[preprint,fleqn,1p]{elsarticle}

\usepackage{amssymb}
\usepackage{amsmath}
\usepackage{amsthm}

\newtheorem{thm}{Theorem}
\newtheorem{lem}{Lemma}

\newtheorem{de}{Definition}

\newtheorem{note}{Notes}
\newtheorem{prop}{Proposition}
\newtheorem{cor}{Corollary}
\newtheorem{exam}{Example}

\def\rnum#1{\resizebox{0.5em}{\height}{\expandafter{\romannumeral #1}}}
\def\Rnum#1{\resizebox{0.5em}{\height}{\uppercase\expandafter{\romannumeral #1}}}

\journal{Theoretical Computer Science}

\begin{document}

\begin{frontmatter}

\title{On the Properties of Language Classes Defined by Bounded Reaction Automata}

\author[label1]{Fumiya Okubo}
\ead{f.okubo@akane.waseda.jp}
\author[label2]{Satoshi Kobayashi}
\ead{satoshi@cs.uec.ac.jp}
\author[label3]{Takashi Yokomori\corref{cor1}}
\ead{yokomori@waseda.jp}
\cortext[cor1]{Corresponding author}

\address[label1]{Graduate School of Education,
Waseda University, 1-6-1 Nishiwaseda, Shinjuku-ku, Tokyo 169-8050,
Japan}
\address[label2]{Graduate School of Informatics and Engineering,
University of Electro-Communications, 1-5-1 Chofugaoka, Chofu-shi, Tokyo 182-8585, Japan}
\address[label3]{Department of Mathematics,  Faculty of Education and Integrated Arts and Sciences,
Waseda University, 1-6-1 Nishiwaseda, Shinjuku-ku, Tokyo 169-8050,
Japan}

\begin{abstract}
Reaction automata are a formal model that has been introduced to  investigate the computing powers of interactive behaviors of biochemical reactions(\cite{OKY:12}). Reaction automata are language acceptors with multiset rewriting mechanism whose basic frameworks are based on reaction systems introduced in \cite{ER:07a}.  

In this paper we continue the investigation of reaction automata with  a  focus on the formal language theoretic properties of subclasses of reaction automata, called {\it linear-bounded} reaction automata (LRAs) and {\it exponentially-bounded} reaction automata (ERAs).    Besides LRAs, we newly introduce an extended model (denoted by $\lambda$-LRAs) by allowing $\lambda$-moves in the accepting process of reaction, and 
investigate the closure properties of language classes accepted by 
both LRAs and $\lambda$-LRAs.   Further, we  establish  new relationships of language classes accepted by LRAs and by ERAs  with the Chomsky hierarchy. 
The main results include the following : \\
\quad (\,i\,) the class of languages accepted by $\lambda$-LRAs forms an AFL with additional closure properties,\\
\quad (\rnum{2})  any recursively enumerable language can be expressed as  
a homomorphic image of a language accepted by an LRA,\\
\quad (\rnum{3}) the class of languages accepted by ERAs coincides with the class of context-sensitive languages. 
\end{abstract}


\end{frontmatter}

\section{Introduction}
There exist two major categories in the research of mathematical modeling of biochemical reactions.  One is an analytical framework based on ordinary differential equations (ODEs) in which macroscopic behaviors of molecules are formulated as ODEs by means of approximating a massive number of molecules (or molecular concentration) by a continuous quantity. The other is a discrete framework based on the multiset rewriting in which a set of various sorts of molecular species in small quantities is represented by a multiset and a biochemical reaction is simulated by replacing the multiset  with another one, under a prescribed condition (\cite{AV:11, CPRS:01, KMP:01, KTZ:09b, Set:01}). 

Among many models that have been investigated from the 
viewpoint of the latter category mentioned above,  Ehrenfeucht and Rozenberg have introduced a formal model called {\it reaction systems} for investigating interactive behaviors between biochemical reactions 
 in which two basic components (reactants and inhibitors) play a key role as a regulation mechanism in controlling biochemical functionalities 
(\cite{ER:07a,ER:07b,ER:09}).   In the same framework, they also introduced the notion of time into reaction systems and investigated notions such as reaction times, creation times of compounds and so forth.  Rather recent two papers  \cite{EMR:10,EMR:11} continue the investigation of reaction systems, with the focuses on combinatorial properties of functions defined by random reaction systems and on the dependency relation between the power of defining functions and  the amount of available resource.    
In the theory of reaction systems, a biochemical reaction is formulated as a triple $a=(R_a, I_a, P_a)$, where $R_a$ is the set of molecules called {\it reactants}, $I_a$ is the set of molecules called {\it inhibitors}, and $P_a$ is the set of molecules called {\it products}. Let $T$ be a set of molecules, and the result of applying a 
 reaction $a$ to $T$, denoted by $res_a(T)$, is given by $P_a$ 
if $a$ is enabled by $T$ (i.e., if $T$ completely includes $R_a$ 
and excludes $I_a$). Otherwise, the result is empty.  Thus, $res_a(T)=R_a$ if $a$ is enabled on $T$, and $res_a(T)=\emptyset$ otherwise. The result of applying a reaction $a$ is extended to the set of reactions $A$, denoted by $res_A(T)$, and an interactive process 
consisting of a sequence of $res_A(T)$'s is properly introduced and investigated. 

Inspired by the works of reaction systems,  we have introduced in \cite{OKY:12} computing devices called {\it reaction automata} and  showed  that they are computationally universal by proving that any recursively enumerable language is accepted by a reaction automaton.  
The notion of reaction automata may be regarded as an extension of reaction systems in the sense that our reaction automata deal with {\it multisets} rather than (usual) sets  as reaction systems do, in the sequence of computational process.  
However,  reaction automata are introduced as computing devices that accept the sets of {\it string objects} (i.e., languages over an alphabet).  This feature of a string accepting device based on multiset 
computing can be realized by introducing a simple idea of feeding an input to the device from the environment and by employing a special encoding technique.

In reaction systems, a number of working assumptions are adopted among which  there are two to be remarked : Firstly, the {\it threshold supply} of elements (molecules) requires that for each element, either enough quantity of it is always supplied to react or it is not present at all. (Thus, reaction systems work with sets rather than multisets.) Secondly, the {\it non-permanency of elements} means that any element not involved in the active reaction ceases  to exist. (Thus, each element has a limited life-span of the unit time.)  
In contrast,  reaction automata assume  properties rather orthogonal to those features of reaction systems: They are defined as computing devices that deal with multisets (rather than sets) in the computing  process of biochemical interactions. It is also assumed that each element is sustained for free if it is not invloved in the reaction.

Before introducing the formal definition of reaction automata in the later  section, we want to describe with an example how a reaction automaton behaves in an interactive way with a given input.  
Figure \ref{graphic} illustrates an intuitive idea of the behavior of a reaction automata $\mathcal {A}_0=(S,\Sigma,A,p_0,f)$, where 
 $S=\{p_0, a, b, c, a', b', c', f \}$ is the set of objects with the input alphabet $\Sigma=\{a, b, c\}$, 
$A=\{{\bf a}_1 = ( a, \{b, b'\}, a' ), \,{\bf a}_2 = ( a'b, \{c, c'\}, b' ),\, {\bf a}_3 = ( b'c, \emptyset, c' ), \, {\bf a}_4 = ( p_0, \{a,b,c,a',b'\}, f )\}$ is the set of reactions, 
 $p_0$ is the initial multiset, and $f$ is the special object to indicate a final multiset. 
Note that in a reaction ${\bf a}=(R_a, I_a, P_a)$, multisets $R_a$ and $P_a$ are  represented by string forms, while $I_a$ is given as a set.  
In  Figure \ref{graphic},  each reaction ${\bf a}_i$ 
is applied to a multiset (of a test tube) after receiving an input symbol (if any is provided from the environment). We assume in this example that no input symbol being fed implies that the input has already been completed. Thus, for instance, when ${\bf a}_4$ is applied to the initial multiset $\{p_0\}$ without any input symbol being fed, which implies that 
the input is the empty string $\lambda$ and that it is  accepted by $\mathcal{A}_0$. For each $n\geq 0$, let $T_n=\{p_0, a'^n\}$. Then,  reactions ${\bf a}_1$ and ${\bf a}_2$ are enabled by the multiset $T_n$ only when inputs $a$ and $b$ are received,  which result in producing $T_{n+1}=\{p_0,a'^{n+1}\}$ and $T'_n=\{p_0,a'^{n-1},b'\}$, respectively.   
Once applying ${\bf a}_2$ to $T_n$ has brought $b'$ into $T'_n$, $\mathcal{A}_0$ has no possibility of applying ${\bf a}_1$ furthermore, because of its inhibitor $\{b,b'\}$.  Afterwards, a successful reaction process can continue only when either $b$ or $c$ is fed, and only an input sequence of $b^{n-1}$ followed by $c^n$ 
eventually leads the reaction process to a multiset $T"_n=\{f,c'^{n}\}$ from which no further multiset is derived and this reaction process terminates. 
One may easily see that $\mathcal{A}_0$ accepts the language $L=\{a^nb^nc^n\mid n\geq 0\}$. One important assumption we would like to remark is that reaction automata allow a multiset of reactions $\alpha$ to apply to a multiset of objects $T$ in an exhaustive manner (what is called  {\it maximally parallel manner}),  and  the interactive process sequence of computation is nondeterministic in that the reaction result from  $T$ may produce more than one product. The details are formally described in the sequel.

\begin{figure}[t]
\centerline{
\includegraphics[scale=0.27]{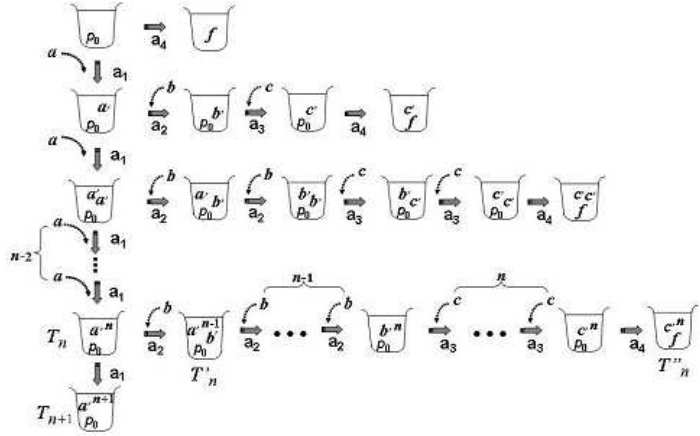}}
\caption{A graphic illustration of interactive biochemical reaction processes for accepting  the language $L=\{a^nb^nc^n \mid n\geq 0\}$ in terms of the reaction automaton  $\mathcal{A}_0$. }
\label{graphic}
\end{figure}

In this paper we  continue the  investigation of reaction automata   
with a  focus on the formal language theoretic properties of subclasses of reaction automata, called {\it linear-bounded} reaction automata (LRAs) and {\it exponentially-bounded} reaction automata (ERAs).  Besides LRAs, we will newly introduce an extended model 
(denoted by $\lambda$-LRAs) by allowing $\lambda$-moves in the accepting process of reaction, and investigate the closure properties of language classes  $\mathcal{LRA}$ and $\lambda$-$\mathcal{LRA}$ 
accepted by  LRAs and $\lambda$-LRAs, respectively.  
 We also investigate the relationships of language classes 
 $\mathcal{LRA}$ and $\mathcal{ERA}$ (the class of languages accepted by ERAs) with the Chomsky hierarchy.

This paper is organized as follows.  After preparing the basic notions and notations from formal language theory in Section 2, we formally describe  the notion of reaction automata (RAs)  and introduce several subclasses of reaction automata such as LRAs, $\lambda$-LRAs and ERAs, based on their volume (space) complexity in Section 3. Then,  the closure properties of the language classes $\mathcal{LRA}$ and $\lambda$-$\mathcal{LRA}$ 
are investigated in  Section 4 and Section 5,  respectively. It is shown  that $\lambda$-$\mathcal{LRA}$ forms an AFL with some additional closure properties.  In Section 6, we also establish  the relations of language classes  $\mathcal{LRA}$ and $\mathcal{ERA}$ to the classes in the Chomsky hierarhy.  Specifically, we show that  any recursively enumerable language can be expressed as  a homomorphic image of a language in $\mathcal{LRA}$. It is also shown that the language class   $\mathcal{ERA}$ coincides with the class of context-sensitive languages.  
Finally,  concluding remarks as well as future research topics are briefly discussed in Section 7.

\section{Preliminaries}

We assume that the reader is familiar with the basic notions of formal language theory.  For unexplained details, refer  to~\cite{HMU:03}. 

Let $V$ be a finite alphabet. For a set $U \subseteq V$, the cardinality of $U$ is denoted by $|U|$. The set of all finite-length strings over $V$ is denoted by $V^*$. The empty string is denoted by $\lambda$.  For a string $x$ in $V^*$,  $|x|$ denotes the length of $x$, while for a symbol $a$ in $V$ we denote  by  $|x|_a$ the number of occurences of $a$ in $x$.  

A morphism $h: V^* \rightarrow U^*$ such that $h(a) \in U$ for all $a \in V$ is called a \textit{coding}, and it is a \textit{weak coding} if $h(a) \in U \cup \{ \lambda \}$ for all $a \in V$. A weak coding is a \textit{projection} if $h(a) \in \{a, \lambda \}$ for each $a \in V$. 

We use the basic notations and definitions regarding multisets that follow~\cite{CMM:01,KMP:01}.
A {\it multiset} over an alphabet $V$ is a mapping $\mu: V \rightarrow \mathbf{N}$, 
where $\mathbf{N}$ is the set of non-negative integers and for each $a \in V$, $\mu(a)$ represents the number of occurrences of $a$ in the multiset $\mu$.  
The set of all multisets over $V$ is denoted by $V^\#$, including the empty multiset denoted by $\mu_{\lambda}$, where $\mu_{\lambda}(a)=0$ for 
all $a\in V$.  A multiset $\mu$ may be represented as a vector, $\mu(V) = (\mu(a_1), \ldots, \mu(a_n))$, for an ordered set $V = \{ a_1, \ldots, a_n \}$.  We can also represent the multiset $\mu$ by any permutation of the string $w_\mu = a^{\mu(a_1)}_1 \cdots a^{\mu(a_n)}_n$. Conversely, with any string $x \in V^*$ one can associate the multiset $\mu_x : V \rightarrow \mathbf{N}$ defined by $\mu_x(a) = |x|_a$ for each $a \in V$. In this sense, we often identify a multiset $\mu$ with its string representation $w_{\mu}$ or any permutation of $w_{\mu}$.  Note that the string representation of $\mu_{\lambda}$ is $\lambda$, i.e., $w_{\mu_{\lambda}}=\lambda$.  

A usual set $U \subseteq V$ is regarded as a multiset $\mu_U$ such that 
 $\mu_U(a) = 1$ if  $a$ is in $U$ and $\mu_U(a)=0$ otherwise.  In particular, 
 for each symbol $a \in V$, a multiset $\mu_{\{a\}}$ is often denoted by $a$ itself.

For two multisets $\mu_1$, $\mu_2$ over $V$,  we define one relation and three operations as follows: 
\[
\begin{array}{ll}
{\it Inclusion}: &\mu_1 \subseteq \mu_2 \text{ \  iff \ } \mu_1(a) \le \mu_2(a),  \text{ for each }  a \in V, \\
\text{{\it Proper inclusion}}: &\mu_1 \subset \mu_2 \text{ \  iff \ } \mu_1 \subseteq \mu_2 \text{ and }  \mu_1 \ne \mu_2, \\
{\it Sum}:&(\mu_1 + \mu_2) (a) = \mu_1(a) + \mu_2(a),  \text{ for each }  a \in V,\\
{\it Intersection}:&(\mu_1 \cap \mu_2) (a) = {\rm min}\{\mu_1(a), \mu_2(a)\},  \text{ for each } a \in V,\\
{\it Difference}:&(\mu_1 - \mu_2) (a) = \mu_1(a) - \mu_2(a),  \text{ for each }  a \in V \text{ (for the case } \mu_2 \subseteq \mu_1 \text{)}.
\end{array}
\]
A multiset $\mu_1$ is called a {\it multisubset} of $\mu_2$ if $\mu_1 \subseteq \mu_2$.  
The sum for a family of multisets $\mathcal{M} = \{\mu_i \}_{i \in I}$ is also denoted by $\sum_{i \in I}\mu_i$. For a multiset $\mu$ and $n \in \mathbf{N}$, $\mu^n$ is defined by $\mu^n (a) = n \cdot \mu(a)$ for each $a \in V$. The {\it weight} of a multiset $\mu$ is $|\mu| = \sum_{a \in V} \mu(a)$.

We introduce an injective function $stm : V^* \to V^\#$ that maps a string to a multiset in the following manner: 
\begin{eqnarray*}
\left\{ \begin{array}{ll}
stm(a_1 a_2 \cdots a_n) = a^{2^{n-1}}_1 \cdots a^2_{n-1} a_n  & \mbox{(for $n\geq 1$)} \\
stm(\lambda) = {\lambda}. &  \\
\end{array} \right.
\end{eqnarray*}

Let us denote by $\mathcal{REG}$ (resp. $\mathcal{LIN, CF, CS, RE}$) 
 the class of regular (resp. linear context-free, context-free, context-sensitive, recursively enumerable) languages.

\section{Reaction Automata and Bounded Variants}

Inspired by the works  of reaction systems, we have introduced the notion of reaction automata in \cite{OKY:12} by extending sets in each reaction to multisets. Here, we start by recalling basic notions concerning reaction automata and their restricted variants called {\it bounded reaction automata}.

\begin{de}
{\rm 
For a set $S$,  a {\it reaction} in $S$ is a 3-tuple ${\bf a} = (R_{\bf a}, I_{\bf a}, P_{\bf a})$ of finite multisets, such that $R_{\bf a}, P_{\bf a} \in S^\#$, $I_{\bf a} \subseteq S$ and $R_{\bf a} \cap I_{\bf a} = \emptyset$.
}
\end{de}
The multisets $R_{\bf a}$ and $P_{\bf a}$ are called the {\it reactant} of ${\bf a}$ and the {\it product} of ${\bf a}$, respectively, while the set $I_{\bf a}$ is called the {\it inhibitor} of ${\bf a}$. These notations are extended to a multiset of reactions as follows:    For a set of reactions $A$ and a multiset $\alpha$ over $A$,  
\[ R_{\alpha} = \sum_{{\bf a}\in A} R_{\bf a}^{\alpha({\bf a})}, \, \, I_{\alpha} = \bigcup_{ {\bf a} \subseteq \alpha} I_{\bf a}, \, P_{\alpha} = \sum_{{\bf a}\in A} P_{\bf a}^{\alpha({\bf a})}. \] 

In what follows,  we usually identify the set of reactions $A$ with the set of labels of reactions in $A$, and often use the symbol $A$ as a finite alphabet.

\begin{de}
{\rm 
Let $A$ be a set of reactions in $S$ and  $\alpha \in A^\#$ be a multiset of reactions 
over $A$.  Then, for  a finite multiset $T \in S^\#$, we say that \\
(1) $\alpha$ is {\it enabled by} $T$ if $R_\alpha \subseteq T$ and $I_\alpha \cap T = \emptyset$, \\
(2)  $\alpha$ is {\it enabled by $T$ in maximally parallel manner}   
if there is no  $\beta \in A^\#$ such that $\alpha \subset \beta$,  and  $\alpha$ and $\beta$  are enabled by  $T$.   \\
(3)  By $En^p_A(T)$ we denote the set of all multisets of reactions $\alpha \in A^\#$ which are enabled by $T$ in maximally parallel manner.\\
(4) The {\it results of $A$ on $T$}, denoted by $Res_{A}(T)$, is defined as follows: 
\[ 
Res_{A}(T) = \{ T - R_{\alpha} + P_{\alpha} \, | \, \alpha \in En^p_A(T)  \}. 
\] 
Note that we have $Res_{A}(T)= \{ T \}$ if $En^p_A(T)=\emptyset$. Thus, if no multiset of reactions $\alpha \in A^{\#}$ is enabled by $T$ in maximally parallel manner, then $T$ remains unchanged.  
}
\end{de}

\begin{note}
{\rm 
 (\,i\,)\ As is mentioned earlier, the definition of the results of $A$ on $T$ given in (4) is in contrast to the original one in \cite{ER:07a},  because 
we  adopt the assumption that  
any element that is not a reactant for any active reaction {\it does} remain in the result after the reaction.\\
(\rnum{2})\ In general, $En^p_A(T)$ may contain more than one element, and therefore, so may $Res_A(T)$.\\
(\rnum{3})\ For simplicity, $I_a$ is often represented as a string rather than a set.
}
\end{note}

We are now in a position to introduce the notion of reaction automata.

\begin{de}{\rm 
{(Reaction Automata)}\ A {\it reaction automaton} (RA) $\mathcal{A}$ is a 5-tuple $\mathcal{A} = (S, \Sigma, A, D_0, f)$, where
\begin{itemize}
\item $S$ is a finite set,  called the {\it background set of}  $\mathcal{A}$,
\item $\Sigma (\subseteq S)$ is called the {\it input alphabet of}  $\mathcal{A}$, 
\item $A$ is a finite set of reactions in $S$,
\item $D_0 \in S^\#$ is an {\it initial multiset},
\item $f \in S$ is a special symbol which indicates the final state.
\end{itemize}
}
\end{de}

\begin{de}{\rm 
Let $\mathcal{A} = (S, \Sigma, A, D_0, f)$ be an RA and $w = a_1 \cdots a_n \in \Sigma^*$.  An {\it interactive process in $\mathcal{A}$ with input $w$} is an infinite sequence  
$\pi = D_0, \ldots, D_i, \ldots$, where 
\begin{eqnarray*}
\left\{ \begin{array}{ll}
 D_{i+1} \in Res_A(a_{i+1}+D_i)   & \mbox{(for $0\leq i \leq n-1$), and} \\
 D_{i+1} \in Res_A(D_i) & \mbox{(for all $i\geq n$)}.
\end{array} \right.
\end{eqnarray*}
In order to represent an interactive process $\pi$, we also use 
the ``arrow notation'' for $\pi$ : 
 $D_0 \rightarrow^{a_1}  D_1 \rightarrow^{a_2} D_2 \rightarrow^{a_3} 
\cdots \rightarrow^{a_{n-1}}  D_{n-1} \rightarrow^{a_{n}}  D_{n}  \rightarrow D_{n+1} \rightarrow \cdots$.   By $IP(\mathcal{A}, w)$ we denote the set of all interactive processes in $\mathcal{A}$ with input $w$.}
\end{de}

For an interactive process $\pi$ in $\mathcal{A}$ with input $w$, if $En^p_A(D_m) = \emptyset$ for some $m \ge |w|$, then we have that  $Res_A(D_m)= \{ D_m \}$ and $D_m =D_{m+1}=
\cdots$. In this case, considering the smallest $m$, we say that $\pi$ {\it converges on} $D_m$ 
(at the $m$-th step). If an interactive process $\pi$ converges 
on $D_m$, then $D_m$ is called the {\it converging state} of $\pi$ and each $D_i$ of $\pi$ is omitted for $i \ge m+1$.

\begin{de}{\rm 
Let $\mathcal{A} = (S, \Sigma, A, D_0, f)$ be an RA. Then, we define:  
\begin{align*}
IP^a(\mathcal{A},w) = \{ \pi \in IP(\mathcal{A}, w) \mid & \mbox{ $\pi$ converges on $D_m$ at the $m$-th step for some 
$m\geq |w|$} \\
  & \mbox{ and $f \subseteq D_m$} \}.  
\end{align*}
The {\it language accepted by} $\mathcal{A}$, denoted by $L(\mathcal{A})$, is defined as follows:
\begin{align*}
L(\mathcal{A}) = \{ w \in \Sigma^* \, | \, & \,  IP^a(\mathcal{A}, w) \ne \emptyset \}. 
\end{align*}}

\end{de}

Let $\mathcal{A}$ be an RA.  
Motivated by the notion of a workspace for a phrase-structure grammar (\cite{AS:73}), we define: for $w\in L(\mathcal{A})$ with $n=|w|$, and for $\pi$ in $IP^a(\mathcal{A},w)$,
\[
WS(w,\pi) = \underset{i}{{\rm max}} \{|D_i| \mid  D_i \mbox{ appears in } \pi \ \}. 
\]
Further, the {\it workspace of} $\mathcal{A}$ {\it for} $w$ is defined as:
\[
WS(w,\mathcal{A}) =\underset{\pi}{{\rm min}} \{WS(w,\pi)  \mid \pi \in IP^a(\mathcal{A},w)\  \}.  
\]
 
\begin{de}{\rm Let  $s$ be a function defined on $\mathbf{N}$.\\
(1)\ An RA  $\mathcal{A}$ is {\it $s(n)$-bounded} if for any $w\in L(\mathcal{A})$ with $n=|w|$,   $WS(w,\mathcal{A})$ is bounded by $s(n)$. \\
(2)\ If a function $s(n)$ is a constant $k$ (resp. linear, polynomial,  exponential), 
then  $\mathcal{A}$ is termed $k$-bounded (resp. linearly-bounded,  polynomially-bounded,  exponentially-bounded), and denoted by 
$k$-RA (resp. $lin$-RA, $poly$-RA, $exp$-RA). Further, 
the class of languages accepted by $k$-RAs (resp. $lin$-RAs,  $poly$-RAs, $exp$-RAs, arbitrary RAs) is denoted by $k$-$\mathcal{RA}$ 
(resp. $\mathcal{LRA, PRA, ERA, RA}$). 
}
\end{de}

\begin{prop}{\rm (Theorem 3 in \cite{OKY:12})} The following inclusions hold {\rm :}  \\
{\rm (1)}. $\mathcal{REG}=k$-$\mathcal{RA} 
\subset  \mathcal{LRA} \subseteq \mathcal{PRA} \subset \mathcal{ERA} \subseteq \mathcal{RA} =  \mathcal{RE}$ {\rm (for each $k\geq 1$)}. \\
{\rm (2)}. $\mathcal{LRA} \subset \mathcal{CS} \subseteq \mathcal{ERA}$. \\
{\rm (3)}. $\mathcal{LIN}$ {\rm (}$\mathcal{CF}${\rm )} and $\mathcal{LRA}$ are incomparable.
\label{prop-ra}
\end{prop}

\begin{exam} \rm{
Let  $\mathcal{A} = (S, \Sigma, A, D_0, f)$ be an LRA defined  as follows:
\begin{align*}
&S = \{ c, p_0, p_1, n_1, c_1, c_2, d, e, f \}, \text{with } \Sigma = \{ c \}, \\
&A = \{ {\bf a}_1, {\bf a}_2, {\bf a}_3, {\bf a}_4, {\bf a}_5, {\bf a}_6, {\bf a}_{7}, {\bf a}_{8} \}, \ \mbox{where} \\
&\quad  {\bf a}_1 = ( p_0, c, p_1),\ \,{\bf a}_2 = ( p_1, ef, p_1 n_1 ), \\
&\quad  {\bf a}_3 = ( c, p_1, c_1 ),\ \, {\bf a}_4 = ( {c^2_1}, p_0 c_2 e, c_2 ), \ \,{\bf a}_5 = ( {c^2_2}, p_0 c_1 e, c_1 ), \\
&\quad  {\bf a}_6 = ( c_1 d, p_0 c_2, e ),\ \,{\bf a}_{7} = ( c_2 d, p_0 c_1, e ),\ \, {\bf a}_{8} = ( e, p_0 c c_1 c_2, f ),   \\
&D_0 = dp_0.
\end{align*}
Then, it holds that $L(\mathcal{A}) = \{ c^{2^n} \, | \, n \ge 0 \}$. Figure \ref{fig-c2n} illustrates the interactive process in $\mathcal{A}$ with the input $c^8$. } \label{exam-c2n}
\end{exam}

\begin{figure}[t]
\centerline{
\includegraphics[scale=0.55]{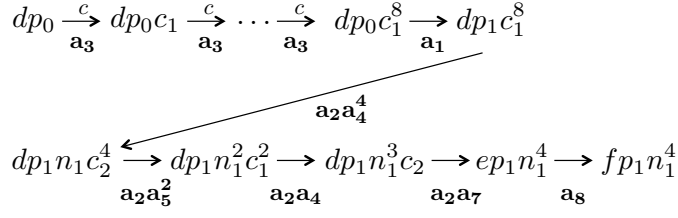}}
\caption{Reaction diagram for accepting $c^8$ in $\mathcal{A}$.}
\label{fig-c2n}
\end{figure}

\section{The closure properties of $\mathcal{LRA}$}

We investigate the closure properties of the class $\mathcal{LRA}$ under various language operations. To this aim, it is convenient to prove the following that one may call {\it normal form lemma} for a bounded class of RAs. 

In what follows, we assume that (i) the symbols (such as $S,\Sigma',S_1,S_2,Q$,etc.) used in the construction for the background set in the proof denote mutually disjoint sets, and (ii) the symbols (such as $p_0,p_1,c,d,f'$,etc.) are newly introduced in the proof.

\begin{de} \rm{
An $s(n)$-bounded RA $\mathcal{A} = (S, \Sigma, A, D_0, f)$ is said to be in} \it{normal form} \rm{if $f$ appears  only in a converging state of an interactive process.}
\end{de}

\begin{lem}
For an $s(n)$-bounded RA $\mathcal{A} = (S, \Sigma, A, D_0, f)$, there exists an $s(n)$-bounded RA $\mathcal{A'} = (S', \Sigma, A', D'_0, f')$ such that $L(\mathcal{A}) = L(\mathcal{A'})$ and $f'$ appears  only in a convergeing state of an interactive process. \label{lem-nf}
\end{lem}

\begin{proof}
For an LRA $\mathcal{A} = (S, \Sigma, A, D_0, f)$, construct an RA $\mathcal{A'} = (S', \Sigma, A', D'_0, f')$ and a mapping $h:S'^{\#} \rightarrow S'^{\#}$ as follows:
\begin{align*}
&S' = S \cup \Sigma' \cup \{ p_0, p_1, c, d, f' \}, \text{ where } \Sigma' = \{ a' | a \in \Sigma \}, \\
&A' = \{ ( h(R), h(I) \cup f', h(P) + c ) \, | \, ( R, I, P ) \in A \} \\
&\qquad \cup \{ ( a, \emptyset, a' ) \, | \, a \in \Sigma \} \cup \{ {\bf a}_1, {\bf a}_2, {\bf a}_3, {\bf a}_4  \}, \ \mbox{where} \\
&\quad  {\bf a}_1 = ( p_0 , \Sigma, p_1 ),\ \,{\bf a}_2 = ( c, \emptyset, \lambda ),\ \, {\bf a}_3 = ( f, cp_0, f' ), \ \,{\bf a}_4 = ( d, \emptyset, h(D_0) ),  \\
&D'_0 = d p_0,
\end{align*}
and
\begin{eqnarray*}
\left\{ \begin{array}{ll}
h(a) = a' \qquad (\text{for }a \in \Sigma),  \\
h(a) = a \qquad \, (\text{for }a \in S' - \Sigma).
\end{array} \right.
\end{eqnarray*}
Let $w \in \Sigma^*$ with $|w|=n$. Then, there exists an interactive process $\pi = D_0, \ldots, D_m \in IP^a(\mathcal{A}, w)$ which converges on $D_m$ if and only if there exists $\pi' = D'_0, \ldots, D'_{m+3} \in IP^a(\mathcal{A'}, w)$ which converges on $D'_{m+3}$  such that
\begin{eqnarray*}
\left\{ \begin{array}{ll|||}
D'_1 = h(D_0) + p_0 + a'_{1}, \\
D'_{i+1} = h(D_i) + p_0 + c^{j_i} + a'_{i+1} \qquad (\text{for } 1 \le i \le n-1, \text{ and some } j_i \ge 1),  \\
D'_{i+1} = h(D_i) + p_1 + c^{k_i} \qquad \qquad \ \ \, (\text{for }  n \le i \le m, \text{ and some } k_i \ge 1), \\
D'_{m+2} = h(D_m) + p_1, \\
D'_{m+3} = h(D_m) - f + f' p_1.
\end{array} \right.
\end{eqnarray*}
Note that (i) $j_i \le | D_{i-1} | \le s(n)$, $k_i \le | D_{i-1} | \le s(n)$. (ii) there may be $D_i$ in $\pi$ with $f \subseteq  D_i$, $0 \le i \le m-1$, but $f'$ cannot be derived from the corresponding state $D'_{i+1}$ in $\pi'$, because the blocking symbol $c$ exists in $D'_{i+2}$. Moreover, the workspace of $\mathcal{A'}$ is obviously $s(n)$-bounded.
\end{proof}

\begin{thm}
$\mathcal{LRA}$ is closed under union, intersection, concatenation, derivative, $\lambda$-free morphisms, $\lambda$-free gsm-mappings and shuffle.
\end{thm}

\begin{proof}
Let $\mathcal{A}_1 = (S_1, \Sigma, A_1, D^{(1)}_{0}, f_1)$ and $\mathcal{A}_2 = (S_2, \Sigma, A_2, D^{(2)}_{0}, f_2)$ be LRAs in normal form with $(S_1 - \Sigma) \cap (S_2 - \Sigma) = \emptyset$. Moreover, let $\Sigma_1 = \{ a^{(1)} \, | \, a \in \Sigma \}$, $\Sigma_2 = \{ a^{(2)} \, | \, a \in \Sigma \}$,  $h_1: {S_1}^{\#} \rightarrow {S_1}^{\#}$ and $h_2: {S_2}^{\#} \rightarrow {S_2}^{\#}$ are defined as follows:
\begin{eqnarray*}
\left\{ \begin{array}{ll}
h_i(a) = a^{(i)} \qquad \, (\text{for }a \in \Sigma),  \\
h_i(a) = a \qquad \quad (\text{for }a \in S_i - \Sigma),
\end{array} \right.
\end{eqnarray*}
for $i \in \{ 1,2 \}$. It is important in the proof of ``union'', ``intersection'', ``concatenation'' and ``shuffle'' parts, that $h_1(S_1)$ and $h_2(S_2)$ are disjoint.

[union] We construct an RA $\mathcal{A} = (S, \Sigma, A, D_0, f)$ as follows:
\begin{align*}
&S = S_1 \cup S_2 \cup \Sigma_1 \cup \Sigma_2 \cup \{ d,f \}, \\
&A = \{ (h_i(R), h_i(I) \cup \{ f \}, h_i(P)) \, | \, ( R, I, P ) \in A_i, i \in \{ 1,2 \} \} \\
&\qquad \cup \{ (a, \emptyset, a^{(1)} a^{(2)} ) \, | \, a \in \Sigma \} \\
&\qquad \cup \{ (f_i, \emptyset, f) \, | \,  i \in \{ 1,2 \} \} \\
&\qquad  \cup \{ (d, \emptyset, h_1(D^{(1)}_{0}) + h_2(D^{(2)}_{0}) \},  \\
&D_0 = d.
\end{align*}
Let $w  = a_1 \cdots a_n$ and let $m = min \{ m_1, m_2 \}$, $m_1,m_2 \ge 0$. Then, for $i = 1$ or $i = 2$, there exists an interactive process $\pi_i = D^{(i)}_{0}, \ldots, D^{(i)}_{m_i} \in IP^a(\mathcal{A}_i, w)$ which converges on $D^{(i)}_{m_i}$ if and only if there exists $\pi = D_0, \ldots, D_{m+2} \in IP^a(\mathcal{A}, w)$ such that
\begin{eqnarray*}
\left\{ \begin{array}{ll}
D_{k+1} = h_1(D^{(1)}_{k}) + h_2(D^{(2)}_{k}) + a^{(1)}_{k+1} a^{(2)}_{k+1} \qquad (\text{for } 0 \le k \le n-1),  \\
D_{k+1} = h_1(D^{(1)}_{k}) + h_2(D^{(2)}_{k}) \qquad \qquad \qquad \ \ \ \, (\text{for } n \le k \le m).
\end{array} \right.
\end{eqnarray*}
(Note that either $D^{(1)}_{m}$ or $D^{(2)}_{m}$ includes $f_1$ and $f_2$, respectively, if and only if $D_{m+2}$ includes $f$.) 

Hence, it holds that $L(\mathcal{A}) = L(\mathcal{A}_1) \cup L(\mathcal{A}_2)$ and the workspace of $\mathcal{A}$ is linear-bounded.

[intersection] In the LRA $\mathcal{A}$ constructed in the proof of ``union'' part,  we replace (i) 
$\{(f_i, \emptyset, f) \, | \,  i \in \{ 1,2 \} \}$ by $\{ (f_1 f_2, \emptyset, f) \}$, and 
(ii) $m= min \{ m_1, m_2 \}$ by $m' = max \{ m_1, m_2 \}$. Then, it is easily seen that  that 
$L(\mathcal{A}) = L(\mathcal{A}_1) \cap L(\mathcal{A}_2)$ holds.

[concatenation] We construct an RA $\mathcal{A} = (S, \Sigma, A, D_0, f)$ as follows:
\begin{align*}
&S = S_1 \cup S_2 \cup \Sigma_1 \cup \Sigma_2 \cup \{ p_1, p_2, d, f \}, \\
&A = \{ (h_i(R), h_i(I) \cup \{ f \}, h_i(P)) \, | \, ( R, I, P ) \in A_i, i \in \{ 1,2 \} \} \\
&\qquad \cup \{ (a, p_2, a^{(1)}) \, | \, a \in \Sigma \}  \cup \{ (d, \emptyset, h_1(D^{(1)}_{0}) ) \} \\
&\qquad \cup \{ (a, p_1, a^{(2)} \, | \, a \in \Sigma \}  \cup \{ (p_1 a, \emptyset, p_2 a^{(2)} + h_2(D^{(2)}_{0}) ) \, | \, a \in \Sigma \} \\
&\qquad \cup \{ (f_1 f_2, \emptyset, f) \},  \\
&D_0 = dp_1.
\end{align*}
Let $w_1, w_2 \in \Sigma^*$ with $|w_1|=n_1$, $|w_2|=n_2$ and $w_1 w_2 = a_1 \cdots a_n$. Then, for $i=1$ and $i=2$, there exists an interactive process $\pi_i = D^{(i)}_{0}, \ldots, D^{(i)}_{m_i} \in IP^a(\mathcal{A}_i, w_i)$ which converges on $D^{(i)}_{m_i}$ if and only if there exists $\pi = D_0, \ldots, D_{m_1 + m_2 +2} \in IP^a(\mathcal{A}, w_1 w_2)$ such that
\begin{eqnarray*}
\left\{ \begin{array}{ll|}
(i) \, D_{k+1} = h_1(D^{(1)}_{k}) + p_1 + a^{(1)}_{k+1} \qquad \qquad \qquad \qquad \ \ \ (\text{for } 0 \le k \le n_1 -1),  \\
(ii) \, D_{k+1} = h_1(D^{(1)}_{k}) + h_2(D^{(2)}_{k - n_1}) + p_2 + a^{(2)}_{k+1} \qquad \ \ \  (\text{for } n_1 \le k \le n-1), \\
(iii) \, D_{k+1} = h_1(D^{(1)}_{k}) + h_2(D^{(2)}_{k - n_1}) + p_2 \qquad \qquad \quad \ \ \, (\text{for } n \le k \le m_1 + m_2).
\end{array} \right.
\end{eqnarray*}
Note that for $D_k$ in $(i)$, a rule in $\{ (a, p_1, a^{(2)}) \, | \, a \in \Sigma \}$ and $\{ (p_1 a, \emptyset, p_2 a^{(2)} + h_2(D^{(2)}_{0}) ) \, | \, a \in \Sigma \}$ is nondeterministically chosen to be applied in the next step. If a rule in $\{ (a, p_1, a^{(2)}) \, | \, a \in \Sigma \}$ is chosen, $D_{k+1}$ is in $(ii)$. 

Hence, it holds that $L(\mathcal{A}) = L(\mathcal{A}_1) \cdot L(\mathcal{A}_2)$ and the workspace of $\mathcal{A}$ is linear-bounded.

[shuffle] We construct an RA $\mathcal{A} = (S, \Sigma, A, D_0, f)$ as follows:
\begin{align*}
&S = S_1 \cup S_2 \cup \Sigma_1 \cup \Sigma_2 \cup \{ d,f \}, \\
&A = \{ (h_i(R), h_i(I) \cup \Sigma_j \cup \{ f \}, h_i(P)) \, | \, ( R, I, P ) \in A_i, i,j \in \{ 1,2 \}, i \ne j \} \\
&\qquad \cup \{ (a, \emptyset, a^{(i)}) \, | \, a \in \Sigma, i \in \{ 1,2 \} \} \\
&\qquad \cup \{ (d, \emptyset, h_1(D^{(1)}_{0}) + h_2(D^{(2)}_{0}) ) \} \cup \{ (f_1 f_2, \emptyset, f) \},  \\
&D_0 = d.
\end{align*}
Let $w_1,w_2 \in \Sigma^*$ with $|w_1|=n_1$, $|w_2|=n_2$ and let $w = a_1 \cdots a_n \in shuf(w_1, w_2)$. Then, for $i = 1$ and $i = 2$, there exists an interactive process $\pi_i = D^{(i)}_{0}, \ldots, D^{(i)}_{m_i} \in IP^a(\mathcal{A}_i, w_i)$ which converges on $D^{(i)}_{m_i}$ if and only if there exists $\pi = D_0, \ldots, D_{m_1 + m_2 +2} \in IP^a(\mathcal{A}, w)$ such that
\begin{eqnarray*}
\left\{ \begin{array}{ll}
D_{k+1} = h_1(D^{(1)}_{k'}) + h_2(D^{(2)}_{k-k'}) + a^{(i)}_{k+1} \qquad \qquad \, (\text{for } 0 \le k \le n-1),  \\
D_{k+1} = h_1(D^{(1)}_{k- n_2}) + h_2(D^{(2)}_{(k - n_1}) \qquad \qquad \qquad (\text{for } n \le k \le m),
\end{array} \right.
\end{eqnarray*}
where $i=1$ or $i=2$ and $0 \le k' \le k$. Note that $i = 1$ ($i = 2$) means that only $\pi_1$ (resp. $\pi_2$) advances to the next step and the value of $k'$ (resp. $k-k'$) is increased by one. 
 
Hence, it holds that $L(\mathcal{A}) = Shuf(L(\mathcal{A}_1), L(\mathcal{A}_2))$ and the workspace of $\mathcal{A}$ is linear-bounded.

[right derivative]  For an LRA $\mathcal{A} = (S, \Sigma, A, D_0, f)$ in normal form and $x = a_1 \cdots a_n \in \Sigma^+$, construct an RA $\mathcal{A'} = (S', \Sigma, A', D'_0, f')$ and a mapping $h:S'^{\#} \rightarrow S'^{\#}$ as follows:
\begin{align*}
&S' = S \cup \Sigma' \cup Q \cup \{ f' \}, \text{ where } \Sigma' = \{ a' | a \in \Sigma \}, Q = \{ q_i \, | \, 0 \le i \le n \}, \\
&A' = \{ ( h(R), h(I) \cup \{ f' \}, h(P) ) \, | \, ( R, I, P ) \in A \} \\
&\qquad \cup \{ ( a, \emptyset, a' ) \, | \, a \in \Sigma \}  \\
&\qquad \cup \{ ( q_i, \Sigma, a'_{i+1} q_{i+1} ) \, | \, 0 \le i \le n-1 \}  \\
&\qquad \cup \{ ( fq_n, \Sigma , f' ) \}, \\
&D'_0 = h(D_0) + q_0,
\end{align*}
and
\begin{eqnarray*}
\left\{ \begin{array}{ll}
h(a) = a' \qquad (\text{for }a \in \Sigma),  \\
h(a) = a \qquad \, (\text{for }a \in S' - \Sigma).
\end{array} \right.
\end{eqnarray*}
Let $wx \in \Sigma^*$ with $w = b_1 \cdots b_l$, $l \ge 1$. Then, there exists an interactive process $\pi = D_{0}, \ldots, D_{m} \in IP^a(\mathcal{A}, wx)$ which converges on $D_{m}$ if and only if there exists $\pi' = D'_0, \ldots, D'_{m + 2} \in IP^a(\mathcal{A}', w)$ such that
\begin{eqnarray*}
\left\{ \begin{array}{ll||}
D'_{k+1} = h(D_{k}) + q_0 b'_{k+1} \qquad \qquad \ \ (\text{for } 0 \le k \le l-1),  \\
D'_{k+1} = h(D_{k}) + q_{k-l+1} a'_{k-l+1} \qquad (\text{for } l \le k \le l+n-1 ), \\
D'_{k+1} = h(D_{k}) + q_n \qquad \qquad \qquad   \ (\text{for } l+n \le k \le m).
\end{array} \right.
\end{eqnarray*}
Hence, it holds that $L(\mathcal{A}) / x = L(\mathcal{A'})$ and the workspace of $\mathcal{A'}$ is linear-bounded.

[left derivative] Let $\mathcal{A} = (S, \Sigma, A, D_0, f)$ be an LRA in normal form and $x = a_1 \cdots a_n \in \Sigma^+$ and $\Sigma_i = \{ a^{(i)} \, | \, a \in \Sigma \}$ for $1 \le i \le n$, $Q = \{ q_i \, | \, 0 \le i \le n \}$.  Construct an RA $\mathcal{A'} = (S', \Sigma, A', D'_0, f)$ and a mapping $h_n :S'^{\#} \rightarrow S'^{\#}$ as follows:
\begin{align*}
&S' = S \cup (\bigcup_{1 \le i \le n}  \Sigma_i) \cup Q \cup \{ d \}, \\
&A' = \{ ( h_n(R), h_n(I), h_n(P) ) \, | \, ( R, I, P ) \in A \} \\
&\qquad \cup \{ ( q_i, \emptyset, a^{(n)}_{i+1} q_{i+1}  ) \, | \, 0 \le i \le n -1 \}  \\
&\qquad \cup \{ ( a, \emptyset, a^{(1)} ) \, | \, a \in \Sigma \}, \\
&\qquad \cup \{ ( a^{(i)}, \emptyset, a^{(i+1)} ) \, | \, a \in \Sigma, 1 \le i \le n -1, n \ge 2 \}, \\
&\qquad \cup \{ ( d, \emptyset, h_n(D_0) ) \}, \\
&D'_0 = dq_0,
\end{align*}
and
\begin{eqnarray*}
\left\{ \begin{array}{ll}
h_n(a) = a^{(n)} \qquad (\text{for }a \in \Sigma),  \\
h_n(a) = a \qquad \quad (\text{for }a \in S' - \Sigma).
\end{array} \right.
\end{eqnarray*}
Let $xw \in \Sigma^*$ with $w = b_1 \cdots b_l$, $l \ge 1$.  Then, there exists an interactive process $\pi = D_{0}, \ldots, D_{m} \in IP^a(\mathcal{A}, xw)$ which converges on $D_{m}$ if and only if there exists $\pi' = D'_0, \ldots, D'_{m + 1} \in IP^a(\mathcal{A}', w)$ such that
\begin{eqnarray*}
\left\{ \begin{array}{ll|}
D'_{k+1} = h_n(D_{k}) + q_{k+1} a^{(n)}_{k+1} b^{(k)}_1 b^{(k-1)}_2 \cdots b^{(1)}_{k}  \qquad \quad \ \  (\text{for } 1 \le k \le n-1 \text{ and } n \ge 2),  \\
D'_{k+1} = h_n(D_{k}) + q_{n} b^{(n)}_{k-n+1} b^{(n-1)}_{k-n} \cdots b^{(1)}_{k} \qquad \qquad \quad (\text{for } n \le k \le l+n -1), \\
D'_{k+1} = h_n(D_{k}) + q_{n}  \qquad \qquad \qquad \qquad \qquad \qquad \quad \,  (\text{for } l+n \le k \le m).
\end{array} \right.
\end{eqnarray*}
Hence, it holds that $x \verb|\| L(\mathcal{A})  = L(\mathcal{A'})$ and the workspace of $\mathcal{A'}$ is linear-bounded.

[$\lambda$-free gsm-mappings] For an LRA $\mathcal{A} = (S, \Sigma, A, D_0, f)$ in normal form and a gsm-mapping $g = (Q, \Sigma, \Delta, \delta, p_0, F)$, construct an RA $\mathcal{A'} = (S', \Delta, A', D'_0, f')$ and a mapping $h: S'^{\#} \rightarrow S'^{\#}$  as follows:
\begin{align*}
&S' = S \cup \Sigma' \cup \Delta \cup Q \cup \{ c, d, f' \}, \text{ where } \Sigma' = \{ a' \, | \, a \in \Sigma \}, \\
&A' = \{ ( h(R), h(I) \cup \{ c,f' \}, h(P) ) \, | \, ( R, I, P ) \in A \} \\
&\qquad \cup \{ ( b, \emptyset, b^2 ) \, | \, b \in \Delta \}  \\
&\qquad \cup \{ ( pc + stm(x), \emptyset, qda ) \, | \, (q,x) \in \delta(p,a) \}  \\
&\qquad \cup \{ ( pd + stm(x), \emptyset, qda ) \, | \, (q,x) \in \delta(p,a), |x|=1 \}  \\
&\qquad \cup \{ ( f''f, \Sigma \cup \{ c,d \}, f' ) \, | \, f'' \in F \} \\
&\qquad \cup \{ ( c, \emptyset, c ) \} \cup \{ ( d, \emptyset, c ) \} \cup \{ ( d, \Sigma, \lambda ) \}, \\
&D'_0 = h(D_0) + cp_0.
\end{align*}
and
\begin{eqnarray*}
\left\{ \begin{array}{ll}
h(a) = a' \qquad (\text{for }a \in \Sigma),  \\
h(a) = a \qquad \, (\text{for }a \in S' - \Sigma).
\end{array} \right.
\end{eqnarray*}

Then, for an input $w= a_1 \cdots a_n$,  there exists  $\pi: D_0, D_1, \ldots, D_m \in  IP^a(\mathcal{A}, w)$ which converges on $D_m$, and  
$g(w)=b_1 \cdots b_{n'}$, where $(p_1, b_1 \cdots b_t) \in \delta(p_0,a_{1})$,   if and only if  there exists the interactive process $\pi'$ in $\mathcal{A}'$ such that
\begin{align*}
D'_0 &\rightarrow^{b_1} h(D_0) + cp_0 b^2_1 \rightarrow^{b_2} \cdots \\
&\rightarrow^{b_{t-1}} h(D_0) + cp_0 + stm(b_1 b_2 \cdots b_{t}) -b_t \\
&\rightarrow^{b_{t}} h(D_0) + dp_0  a'_1 \\
&\rightarrow^{b_{t+1}} h(D_1) + cp_1  b^2_{t+1} \rightarrow^{b_{t+2}} \cdots \\
&(\text{or } \rightarrow^{b_{t+1}} h(D_1) + dp_1 a'_2 \rightarrow^{b_{t+2}} \cdots \text{ if } (q, b_{t+1}) \in \delta(p_1,a_{2})) \\
&\rightarrow^{b_{n'}} h(D_{n-1}) + df'' a'_n \\
&\rightarrow h(D_{n}) + f'' \\
&\rightarrow h(D_{n}) - f +f'(=D'_q)
\end{align*}
and $D'_q$ is a converging state in $\mathcal{A'}$.  
Hence, it holds that $g(L(\mathcal{A})) = L(\mathcal{A'})$ and the workspace of $\mathcal{A'}$ is linear-bounded.

[$\lambda$-free morphisms] Since $\mathcal{LRA}$ is closed under $\lambda$-free gsm-mappings, it is also closed under $\lambda$-free morphisms.
\end{proof}

In order to prove some of the negative closure properties of $\mathcal{LRA}$, the following two lemmas are of crucially importance.

\begin{lem}{\rm (Lemma 1 in \cite{OKY:12})}
For an alphabet $\Sigma$ with $|\Sigma| \ge 2$, let $h:\Sigma^* \rightarrow \Sigma^*$ be an injection such that for any $w \in \Sigma^*$, $|h(w)|$ is bounded by a polynomial of $|w|$. Then, there is no PRA $\mathcal{A}$ such that $L(\mathcal{A}) = \{ wh(w) \, | \,  w \in \Sigma^* \}$. 
\label{lem-ww}
\end{lem}

\begin{lem} 
$L_1 = \{w_1 w_2 \ |  \ w_1, w_2 \in \{ a,b \}^*, \   w_1 \ne w_2  \} \in \mathcal{LRA}$. 
\label{lem-nww}
\end{lem}

\begin{proof}
Let $L = \{ u_1 s u_2 v_1 t v_2 \ | \ u_1, u_2, v_1, v_2 \in \{ a, b \}^*, |u_1|=|v_1|, |u_2|=|v_2|, s,t \in \{ a, b \}, s \ne t \}$ and $\mathcal{A} = (S, \Sigma, A, D_0, f)$ be an LRA defined as follows:
\begin{align*}
&S = \{ a, b, a', b', c_1, c_2, p_0, p_1, p_2, p_3, f \} \mbox{ with }  \Sigma=\{a,b\},  \\
&A = \{ {\bf a}_1, {\bf a}_2, {\bf a}_3, {\bf a}_4, {\bf a}_5, {\bf a}_6, {\bf a}_7, {\bf a}_8, {\bf a}_9, {\bf a}_{10}, {\bf a}_{11}, {\bf a}_{12}, {\bf a}_{13}, {\bf a}_{14}, {\bf a}_{15}  \}, \ \mbox{where} \\
&\quad  {\bf a}_1 = ( p_0 a, \emptyset, p_0 c_1 ),\ \,{\bf a}_2 = ( p_0 b, \emptyset, p_0 c_1 ),\ \, {\bf a}_3 = ( p_0 a, \emptyset, p_1 a' ), \ \,{\bf a}_4 = ( p_0 b, \emptyset, p_1 b' ), \\
&\quad  {\bf a}_5 = ( p_1 a, \emptyset, p_1 c_2 ),\ \,{\bf a}_6 = ( p_1 b, \emptyset, p_1 c_2 ),\ \, {\bf a}_7 = ( p_1 a, \emptyset, p_2 c_2 ), \ \,{\bf a}_8 = ( p_1 b, \emptyset, p_2 c_2 ), \\
&\quad  {\bf a}_9 = ( p_2 a c_1, \emptyset, p_2 ),\ \,{\bf a}_{10} = ( p_2 b c_1, \emptyset, p_2 ),\ \, {\bf a}_{11} = ( p_2 a' b, c_1, p_3 ), \ \, {\bf a}_{12} = ( p_2 b' a, c_1, p_3 ), \\
&\quad  {\bf a}_{13} = ( p_3 a c_2, \emptyset, p_3 ),\ \,{\bf a}_{14} = ( p_3 b c_2, \emptyset, p_3 ),\ \, {\bf a}_{15} = ( p_3, abc_2, f ),  \\
&D_0 = p_0.
\end{align*}
Let $w = u_1 s u_2 v_1 t v_2 \in L$ be an input string. The string $w$ is accepted by $\mathcal{A}$ in the following manner: 
\begin{enumerate}
\item Applying ${\bf a}_1$ and ${\bf a}_2$, the length of $u_1$ is counted by the number of $c_1$. 
\item Applying ${\bf a}_3$ or ${\bf a}_4$, $s$ is rewritten by $s'$. 
\item Applying ${\bf a}_5$, ${\bf a}_6$, ${\bf a}_7$ and ${\bf a}_8$, the length of $u_2$ is counted by the number of $c_2$. If ${\bf a}_7$ or ${\bf a}_8$ is applied, then the interactive process enters the next step. \item Applying ${\bf a}_9$ and ${\bf a}_{10}$, it is confirmed that $u_1=v_1$ by consuming $c_1$. 
\item Applying ${\bf a}_{11}$ and ${\bf a}_{12}$, it is confirmed that $s \ne t$. 
\item Applying ${\bf a}_{13}$ and ${\bf a}_{14}$, it is confirmed that $u_2=v_2$ by consuming $c_2$.
\end{enumerate}

Therefore, it holds that $L = L(\mathcal{A})$. Note that $L_1 = L \cup \{ w \in \Sigma^* \ | \ |w| = 2n+1, n \ge 0 \}$. Since $\mathcal{LRA}$ is closed under union and includes all regular language, $L_1$ is in $\mathcal{LRA}$.
\end{proof}

\begin{thm}
$\mathcal{LRA}$ is not closed under complementation, quotient by regular languages, morphisms or gsm-mappings.
\end{thm}

\begin{proof}
From Lemma \ref{lem-nww}, $L_1 = \{w_1 w_2 \ |  \ w_1, w_2 \in \{ a,b \}^*, \   w_1 \ne w_2  \} \in \mathcal{LRA}$, while from Lemma \ref{lem-ww}, $\bar{L_1} = \{w_1 w_2 \ |  \ w_1, w_2 \in \{ a,b \}^*, \   w_1 = w_2  \} \notin \mathcal{LRA}$. Hence, $\mathcal{LRA}$ is not closed under complementation. From Corollary \ref{cor-re-lra}, it obviously follows that $\mathcal{LRA}$ is not closed under quotient by regular languages, morphisms or gsm-mappings.
\end{proof}

\section{The closure properties of $\lambda$-$\mathcal{LRA}$}

As is seen in the previous section, it remains open whether or not 
the class $\mathcal{LRA}$ is closed under several basic operations 
such as Kleene closures $(+, *)$ or inverse homomorphism. 

In this section, we shall prove that if the $\lambda$-move is allowed in the phase of input mode in the transition process of 
reactions, then  the obtained class of languages ($\lambda$-$\mathcal{LRA}$ introduced below) accepted in that manner shows in turn positive closure properties under  those basic operations.

\begin{de}{\rm 
Let $\mathcal{A} = (S, \Sigma, A, D_0, f)$ be an RA. An interactive process} \it{in the $\lambda$-input mode} \rm{in $\mathcal{A}$ with input $w \in \Sigma^*$ is a sequence  
$\pi = D_0, \ldots, D_i, \ldots$, where $w = a_1 \cdots a_n$ with $a_{i} \in \Sigma \cup \{ \lambda \}$ for $1 \leq i \leq n$,
\begin{eqnarray*}
\left\{ \begin{array}{ll}
 D_{i+1} \in Res_A(a_{i+1}+D_i)   & \mbox{(for $0\leq i \leq n-1$), and} \\
 D_{i+1} \in Res_A(D_i) & \mbox{(for all $i\geq n$)}.
\end{array} \right.
\end{eqnarray*}
By $IP_{\lambda}(\mathcal{A}, w)$ we denote the set of all interactive processes} \it{in the $\lambda$-input mode} \rm{in $\mathcal{A}$ with input $w$.}
\end{de}

\begin{de}{\rm 
Let $\mathcal{A} = (S, \Sigma, A, D_0, f)$ be an RA. Then, we define: 
\begin{align*}
IP^a_{\lambda}(\mathcal{A},w) = \{ \pi \in IP_{\lambda}(\mathcal{A}, w) \mid & \mbox{ $\pi$ converges on $D_m$ at the $m$-th step for some 
$m\geq |w|$} \\
   & \mbox{ and $f \subseteq D_m$} \}.  
\end{align*}
The {\it language accepted by} $\mathcal{A}$ {\it in the $\lambda$-input mode}, denoted by 
$L_{\lambda}(\mathcal{A})$, is defined as follows:
\begin{align*}
L_{\lambda}(\mathcal{A}) = \{ w \in \Sigma^* \, | \, & \,  IP^a_{\lambda}(\mathcal{A}, w) \ne \emptyset \}. 
\end{align*}}
\end{de}

\begin{de}{\rm 
The class of languages accepted by RAs ($k$-RAs, $lin$-RAs,  $poly$-RAs and $exp$-RAs)} \it{in the $\lambda$-input mode} {\rm  is denoted by $\lambda$-$\mathcal{RA}$ 
(resp. $\lambda$-$k$-$\mathcal{RA}$, $\lambda$-$\mathcal{LRA}$, $\lambda$-$\mathcal{PRA}$ and $\lambda$-$\mathcal{ERA}$). 
}
\end{de}

In what follows, we focus on dealing with $\lambda$-$\mathcal{LRA}$ and continue investigating the clouser properties of the class of languages. 
As a result, it is shown that the class forms an AFL, i.e., an abstract family of languages. 

\begin{thm}
For any LRA $\mathcal{A}$, there exists an LRA $\mathcal{A'}$ such that $L(\mathcal{A}) = L_{\lambda}(\mathcal{A'})$.
\end{thm}

\begin{proof}
Let $\Sigma' = \{ a' \, | \, a \in \Sigma \}$ be a new alphabet.
For an LRA $\mathcal{A} = (S, \Sigma, A, D_0, f)$ in normal form, construct an RA $\mathcal{A'} = (S', \Sigma, A', D'_0, f')$ and a mapping $h: S'^{\#} \rightarrow S'^{\#}$ as follows:
\begin{align*}
&S' = S \cup \Sigma' \cup \{ p_0, p_1, d, f' \},  \\
&A' = \{ h(R), h(I) \cup \{ f' \}, h(P) \, | \, ( R, I, P ) \in A \} \cup \{ (a, p_1, a') \, | \, a \in \Sigma \} \\
&\qquad \cup \{ {\bf a}_1, {\bf a}_2, {\bf a}_3 \},\mbox{ where} \\
&\qquad \qquad {\bf a}_1 = ( d, \Sigma, h(D_0) ),\ \,{\bf a}_2 = ( p_0, \Sigma, p_1 ),\ \, {\bf a}_3 = ( f, \Sigma, f' ),  \\
&D'_0 = d p_0,
\end{align*}
and
\begin{eqnarray*}
\left\{ \begin{array}{ll}
h(a) = a' \qquad (\text{for }a \in \Sigma),  \\
h(a) = a \qquad \, (\text{for }a \in S' - \Sigma).
\end{array} \right.
\end{eqnarray*}
Note that once $\lambda$ is inputted before an element $a \in \Sigma$ in an interactive process, $a$ cannot be consumed since $p_1$ will have to be introduced by ${\bf{a}}_2$ in the next step, which implies that no $\lambda$-input is  allowed before an element $a \in \Sigma$ in a successful interactive process in $\mathcal{A}'$.
\end{proof}

\begin{de} \rm{
An $s(n)$-bounded RA $\mathcal{A} = (S, \Sigma, A, D_0, f)$ is said to be in} \it{$\lambda$-normal form} \rm{if $f$ appears  only in a converging state of an interactive process in the $\lambda$-input mode.}
\end{de}

\begin{lem}
For an $s(n)$-bounded RA $\mathcal{A} = (S, \Sigma, A, D_0, f)$, there exists an $s(n)$-bounded RA $\mathcal{A'} = (S', \Sigma, A', D'_0, f')$ such that $L_{\lambda}(\mathcal{A}) = L_{\lambda}(\mathcal{A'})$ and $f'$ appears  only in a converging state of an interactive process.
\end{lem}

\begin{proof}
For an $s(n)$-bounded RA $\mathcal{A} = (S, \Sigma, A, D_0, f)$, construct an RA $\mathcal{A'} = (S', \Sigma, A', D'_0, f')$ and a mapping $h:S'^{\#} \rightarrow S'^{\#}$ as follows:
\begin{align*}
&S' = S \cup \Sigma' \cup \{ p_0, p_1, c, d, f' \}, \text{ where } \Sigma' = \{ a' | a \in \Sigma \}, \\
&A' = \{ ( h(R), h(I) \cup f', h(P) + c ) \, | \, ( R, I, P ) \in A \} \\
&\qquad \cup \{ ( a, p_1, a' ) \, | \, a \in \Sigma \} \cup \{ {\bf a}_1, {\bf a}_2, {\bf a}_3, {\bf a}_4, {\bf a}_5  \}, \ \mbox{where} \\
&\quad  {\bf a}_1 = ( p_0 , \Sigma, p_1 ),\ \,{\bf a}_2 = ( c, \emptyset, \lambda ),\ \, {\bf a}_3 = ( f, \{ c,p_0 \} \cup \Sigma, f' ), \\
&\quad {\bf a}_4 = ( d, \emptyset, h(D_0) ), \ \,{\bf a}_5 = ( p_0, \emptyset, p_0 ),  \\
&D'_0 = d p_0,
\end{align*}
and
\begin{eqnarray*}
\left\{ \begin{array}{ll}
h(a) = a' \qquad (\text{for }a \in \Sigma),  \\
h(a) = a \qquad \, (\text{for }a \in S' - \Sigma).
\end{array} \right.
\end{eqnarray*}
When $\lambda$ is inputted in an interactive process, ${\bf a}_1$ exclusively or ${\bf a}_5$ has to be used in the next step. Using ${\bf a}_1$ implies that the input of the string terminates, while using ${\bf a}_5$ implies that the input of the string continues.
The rest of the key issue is proved in a similar manner to Lemma \ref{lem-nf}.
\end{proof}

\begin{thm}
$\lambda$-$\mathcal{LRA}$ is closed under union, intersection, concatenation, Kleene $+$, Kleene $*$, derivative, $\lambda$-free morphisms, inverse morphisms, $\lambda$-free gsm-mappings and shuffle.
\end{thm}

\begin{proof}

[union, concatenation and shuffle] Using the same construction as the proof of Theorem 1, the claims are immediately proved. 

[intersection] Let $\mathcal{A}_1 = (S_1, \Sigma, A_1, D^{(1)}_{0}, f_1)$ and $\mathcal{A}_2 = (S_2, \Sigma, A_2, D^{(2)}_{0}, f_2)$ be LRAs in $\lambda$-normal form with $(S_1 - \Sigma) \cap (S_2 - \Sigma) = \emptyset$. Moreover, let $\Sigma_i = \{ a^{(i)} \, | \, a \in \Sigma \}$, $\Sigma'_i = \{ a^{(i)'} \, | \, a^{(i)} \in \Sigma_i \}$ be alphabets and $h_i: {S_i}^{\#} \rightarrow {S_i}^{\#}$ be a mapping defined as follows:
\begin{eqnarray*}
\left\{ \begin{array}{ll}
h_i(a) = a^{(i)} \qquad \, (\text{for }a \in \Sigma),  \\
h_i(a) = a \qquad \quad (\text{for }a \in S_i - \Sigma),
\end{array} \right.
\end{eqnarray*}
for $i \in \{ 1,2 \}$. 

Then, we construct an RA $\mathcal{A} = (S, \Sigma, A, D_0, f)$ as follows:
\begin{align*}
&S = S_1 \cup S_2 \cup \Sigma_1 \cup \Sigma_2 \cup \Sigma'_1 \cup \Sigma'_2 \cup \{ d,f \}, \\
&A = \{ (h_i(R), h_i(I) \cup \Sigma \cup \Sigma'_j \cup \{ f \}, h_i(P)) \, | \, ( R, I, P ) \in A_i, i \in \{ 1,2 \} \} \\
&\qquad \cup \{ (a, \Sigma'_1 \cup \Sigma'_2, a^{(1)'} a^{(2)'}) \, | \, a \in \Sigma \} \\
&\qquad \cup \{ (a^{(i)'}, \Sigma, a^{(i)}) \, | \, a^{(i)} \in \Sigma_i, i \in \{ 1,2 \} \} \\
&\qquad \cup \{ (a^{(i)'}, \Sigma, a^{(i)'}) \, | \, a^{(i)} \in \Sigma_i, i \in \{ 1,2 \} \} \\
&\qquad  \cup \{ (d, \emptyset, h_1(D^{(1)}_{0}) + h_2(D^{(1)}_{0}) ) \} \cup \{ (f_1 f_2, \Sigma \cup \Sigma'_1 \cup \Sigma'_2, f) \},  \\
&D_0 = d.
\end{align*}
Let $w  = a_1 \cdots a_{n} \in L_{\lambda}(\mathcal{A}_1) \cap L_{\lambda}(\mathcal{A}_2)$. 
Moreover, let $D^{(1)}_{i} \rightarrow^{\lambda} \cdots \rightarrow^{\lambda} D^{(1)}_{j} \rightarrow^{a_{m}} D^{(1)}_{j+1}$ be a part of $\pi_1 \in IP_{\lambda}(\mathcal{A}_1, w)$ and $D^{(2)}_{k}  \rightarrow^{\lambda} \cdots \rightarrow^{\lambda} D^{(2)}_{l} \rightarrow^{a_{m}} D^{(2)}_{l+1}$ be a part of $\pi_2 \in IP_{\lambda}(\mathcal{A}_2, w)$ for $1 \le m \le n$. We assume that $j-i \le l-k$. Then, they are imitated in $\pi \in IP_{\lambda}(\mathcal{A}, w)$ as follows:
\begin{align*}
&h_1(D^{(1)}_{i})+h_2(D^{(2)}_{k}) \rightarrow^{\lambda} \cdots \rightarrow^{\lambda} h_1(D^{(1)}_{j})+h_2(D^{(2)}_{k-i+j}) \\
\rightarrow^{a_{m}} &h_1(D^{(1)}_{j})+h_2(D^{(2)}_{k-i+j}) + a^{(1)'}_{m} a^{(2)}_{m} \\  
\rightarrow^{\lambda} \ \, &h_1(D^{(1)}_{j})+h_2(D^{(2)}_{k-i+j}) + a^{(1)}_{m} a^{(2)'}_{m} \\
\rightarrow^{\lambda} \ \, &h_1(D^{(1)}_{j})+h_2(D^{(2)}_{k-i+j+1}) + a^{(1)}_{m} a^{(2)'}_{m} \rightarrow^{\lambda} \cdots \\
\rightarrow^{\lambda} \ \, &h_1(D^{(1)}_{j})+h_2(D^{(2)}_{k}) + a^{(1)}_{m} a^{(2)}_{m} \\
\rightarrow^{\lambda} \ \, &h_1(D^{(1)}_{j+1})+h_2(D^{(2)}_{k+1}).
\end{align*}

The other direction of the proof is shown in the similar manner.
Hence, it holds that $L_{\lambda}(\mathcal{A}) = L_{\lambda}(\mathcal{A}_1) \cap L_{\lambda}(\mathcal{A}_2)$ and the workspace of $\mathcal{A}$ is linear-bounded.

[Kleene $*$] Let $\mathcal{A} = (S, \Sigma, A, D_0, f)$ be an LRA in $\lambda$-normal form and $\Sigma' = \{ a' \, | \, a \in \Sigma \}$. Construct an RA $\mathcal{A'} = (S', \Sigma, A', D'_0, f')$ and a mapping $h :S'^{\#} \rightarrow S'^{\#}$ as follows:
\begin{align*}
&S' = S \cup \Sigma' \cup \{p_0, p_1, d, e, f'', f' \}, \\
&A' = \{ ( h(R), h(I) \cup \{ f', f'' \}, h(P) ) \, | \, ( R, I, P ) \in A \} \\
&\qquad \cup \{ ( h(R), h(I) \cup \Sigma \cup \{ f', f'' \}, h(P) ) \, | \, ( R, I, P ) \in A, \, f \subseteq P \} \\
&\qquad \cup \{ ( a, p_1, a'  ) \, | \, a \in \Sigma \}  \\
&\qquad \cup \{ ( a, e,  \lambda ) \, | \, a \in (S \cup \Sigma') - \Sigma  \} \\
&\qquad \cup \{ ( d, \emptyset, h(D_0) + e ) \} \cup \{ ( p_0, \Sigma,  p_1) \} \cup \{ ( p_0, \emptyset,  p_0) \}  \\
&\qquad \cup \{ ( ef, \Sigma, f'' ) \} \cup \{ ( f'', \{ p_1 \},  h(D_0) + e) \} \cup \{ ( f'', \Sigma \cup \{ p_0 \},  f') \},  \\
&D'_0 = d p_0,
\end{align*}
and
\begin{eqnarray*}
\left\{ \begin{array}{ll}
h(a) = a' \qquad (\text{for }a \in \Sigma),  \\
h(a) = a \qquad \, (\text{for }a \in S' - \Sigma).
\end{array} \right.
\end{eqnarray*}
Let $w_1, w_2 \in \Sigma^*$ with $w_1 = a^{(1)}_{1} \cdots a^{(1)}_{n}, w_2 = a^{(2)}_{1} \cdots a^{(2)}_{m}$.
Then, we can easily see that there exist the interactive processes $D_0, D^{(1)}_{1}, \ldots D^{(1)}_{i} \in IP^a_{\lambda}(\mathcal{A}, w_1)$ and $D_0, D^{(2)}_{1} \ldots D^{(2)}_{j} \in IP^a_{\lambda}(\mathcal{A}, w_2)$ which converge on $D^{(1)}_i$ and $D^{(2)}_j$, respectively, if and only if there exists the interactive process $D'_0, D'_1, \ldots D'_{i+j+4} \in IP^a_{\lambda}(\mathcal{A}', w_1w_2)$ such that
\begin{eqnarray*}
\left\{ \begin{array}{ll|}
D'_{k+1} = h(D^{(1)}_{k}) + ep_0  \qquad \qquad \quad  (\text{for } 0 \le k \le i),  \\
D'_{i+2} = h(D^{(1)}_{i}) - f + f''p_0 ,  \\
D'_{i+3} = h(D_{0}) + ep_0,  \\
D'_{k+i+4} = h(D^{(2)}_{k}) + ep_0  \qquad \qquad \,  (\text{for } 0 \le k \le j),  \\
D'_{i+j+5} = h(D^{(2)}_{j}) - f + f''p_1 ,  \\
D'_{i+j+6} = fp_1  \\
\end{array} \right.
\end{eqnarray*}
Hence, it holds that $w_1w_2 \in L_{\lambda}(\mathcal{A}')$.
In a similar manner, we can prove that $w_1 \cdots w_l \in L_{\lambda}(\mathcal{A}')$ for $w_1, \ldots, w_l \in L_{\lambda}(\mathcal{A})$ and $l \ge 0$.
Then, it holds that $L_{\lambda}(\mathcal{A})^*  = L_{\lambda}(\mathcal{A'})$ and the workspace of $\mathcal{A'}$ is linear-bounded.

[Kleene $+$] For LRAs $\mathcal{A}$ and $\mathcal{A'}$ in the proof of ``Kleene $*$'' part, it holds that $L_{\lambda}(\mathcal{A})^+  = L_{\lambda}(\mathcal{A'}) \cap \Sigma^+$.  Since $\lambda$-$\mathcal{LRA}$ is closed under intersection with regular languages, it is also closed under  Kleene $+$.

[right derivative]  For an LRA $\mathcal{A} = (S, \Sigma, A, D_0, f)$ in $\lambda$-normal form and $x = a_1 \cdots a_n \in \Sigma^+$, construct an RA $\mathcal{A'} = (S', \Sigma, A', D'_0, f')$ and a mapping $h:S'^{\#} \rightarrow S'^{\#}$ as follows:
\begin{align*}
&S' = S \cup \Sigma' \cup Q \cup \{ f' \}, \text{ where } \Sigma' = \{ a' | a \in \Sigma \}, Q = \{ q_i \, | \, 0 \le i \le n \}, \\
&A' = \{ ( h(R), h(I) \cup \{ f' \}, h(P) ) \, | \, ( R, I, P ) \in A \} \\
&\qquad \cup \{ ( a, Q - \{ q_0 \}, a' ) \, | \, a \in \Sigma \}  \\
&\qquad \cup \{ ( q_i, \Sigma, a'_{i+1} q_{i+1} ) \, | \, 0 \le i \le n-1 \}  \\
&\qquad \cup \{ ( q_i, \Sigma, q_i ) \, | \, 0 \le i \le n \}  \\
&\qquad \cup \{ ( fq_n, \Sigma , f' ) \}, \\
&D'_0 = h(D_0) + q_0,
\end{align*}
and
\begin{eqnarray*}
\left\{ \begin{array}{ll}
h(a) = a' \qquad (\text{for }a \in \Sigma),  \\
h(a) = a \qquad \, (\text{for }a \in S' - \Sigma).
\end{array} \right.
\end{eqnarray*}
Note that because of the inhibitor of a reaction in $\{ ( a, Q - \{ q_0 \}, a' ) \, | \, a \in \Sigma \}$, a reaction in $\{ ( q_i, \Sigma, a'_{i+1} q_{i+1} ) \, | \, 0 \le i \le n-1 \}$ must be used after feeding the input. 
Hence, the rest of the proof is similar to the case for the ordinary input mode.

[left derivative] For an LRA $\mathcal{A} = (S, \Sigma, A, D_0, f)$ in $\lambda$-normal form and $x = a_1 \cdots a_n \in \Sigma^+$, construct an RA $\mathcal{A'} = (S', \Sigma, A', D'_0, f')$ and a mapping $h:S'^{\#} \rightarrow S'^{\#}$ as follows:
\begin{align*}
&S' = S \cup \Sigma' \cup Q \cup \{ f' \}, \text{ where } \Sigma' = \{ a' | a \in \Sigma \}, Q = \{ q_i \, | \, 0 \le i \le n \}, \\
&A' = \{ ( h(R), h(I) \cup \{ f' \}, h(P) ) \, | \, ( R, I, P ) \in A \} \\
&\qquad \cup \{ ( a, Q - \{ q_n \}, a' ) \, | \, a \in \Sigma \}  \\
&\qquad \cup \{ ( q_i, \Sigma, a'_{i+1} q_{i+1} ) \, | \, 0 \le i \le n-1 \}  \\
&\qquad \cup \{ ( q_i, \Sigma, q_i ) \, | \, 0 \le i \le n \}  \\
&\qquad \cup \{ ( fq_n, \Sigma , f' ) \}, \\
&D'_0 = h(D_0) + q_0,
\end{align*}
and
\begin{eqnarray*}
\left\{ \begin{array}{ll}
h(a) = a' \qquad (\text{for }a \in \Sigma),  \\
h(a) = a \qquad \, (\text{for }a \in S' - \Sigma).
\end{array} \right.
\end{eqnarray*}
Note that because of the inhibitor of a reaction in $\{ ( a, Q - \{ q_n \}, a' ) \, | \, a \in \Sigma \}$, each reaction in $\{ ( q_i, \Sigma, a'_{i+1} q_{i+1} ) \, | \, 0 \le i \le n-1 \}$ must be used before starting the input except $\lambda$.

Let $xw \in \Sigma^*$ with $w = b_1 \cdots b_l$. Then, there exists an interactive process $\pi = D_{0}, \ldots, D_{m} \in IP^a_{\lambda}(\mathcal{A}, xw)$ which converges on $D_{m}$ if and only if there exists $\pi' = D'_0, \ldots, D'_{m + 2} \in IP^a_{\lambda}(\mathcal{A}', w)$ such that
\begin{eqnarray*}
\left\{ \begin{array}{ll||}
D'_{i+1} = h(D_{i}) + q_{k+1} a'_{k+1} \qquad \qquad (\text{for } 0 \le k \le n-1),  \\
D'_{i+1} = h(D_{i}) + q_{n} b'_{k-n+1} \qquad \qquad (\text{for } n \le k \le l+n-1 ), \\
D'_{i+1} = h(D_{i}) + q_n, 
\end{array} \right.
\end{eqnarray*}
for some $0 \le i \le m$.
Hence, it holds that $x \backslash L_{\lambda}(\mathcal{A}) = L_{\lambda}(\mathcal{A'})$ and the workspace of $\mathcal{A'}$ is linear-bounded.

[inverse morphisms] Let $\mathcal{A} = (S, \Delta, A, D_0, f)$ be an LRA in normal form and $h:\Sigma^* \rightarrow \Delta^*$ be a morphism defined as $h(a) = b_{(a,1)} \cdots b_{(a, l)} \in \Delta^*$ or $h(a) = \lambda$, for $a \in \Sigma$ and $|h(a)|=l$. Moreover, let $\Delta' = \{ b' \, | \, b \in \Delta \}$ and $Q = \{ q_{(a, i)} \, | \, a \in \Sigma, \,  |h(a)| \ge 2, \, 1 \le i \le |h(a)| - 1 \}$.  Construct an RA $\mathcal{A'} = (S', \Sigma, A', D'_0, f')$ and a mapping $g :S'^{\#} \rightarrow S'^{\#}$ as follows:
\begin{align*}
&S' = S \cup \Sigma \cup \Delta' \cup Q \cup \{ d, f' \}, \\
&A' = \{ ( g(R), g(I) \cup \{ a \in \Sigma \, | \, |h(a)| = 0 \} \cup \{ f' \}, g(P) ) \, | \, ( R, I, P ) \in A \} \\
&\qquad \cup \{ ( a, \emptyset, \lambda  ) \, | \, |h(a)| = 0, a \in \Sigma \}  \\
&\qquad \cup \{ ( a, \emptyset, b'_{(a,1)} ) \, | \, |h(a)| = 1, a \in \Sigma \} \\
&\qquad \cup \{ ( a, \emptyset, q_{(a, 1)} b'_{(a,1)} ), ( q_{(a, 1)}, \Sigma, b'_{(a,2)} ) \, | \, |h(a)| = 2, a \in \Sigma \} \\
&\qquad \cup \{ ( a, \emptyset, q_{(a, 1)} b'_{(a,1)} ), ( q_{(a, i)}, \Sigma, q_{(a, i+1)} b'_{(a, i+1)} ), ( q_{(a, |h(a)| - 1)}, \Sigma, b'_{(a, |h(a)|)} ) \, \\
&\qquad \qquad \qquad \qquad \qquad \qquad \qquad \quad  | \, |h(a)| \ge 3, \, 1 \le i \le |h(a)| - 2, a \in \Sigma \} \\
&\qquad \cup \{ ( q, \emptyset, q ) \, | \, q \in Q \} \\
&\qquad \cup \{ ( d, \emptyset, g(D_0) ) \} \cup \{ ( f, Q \cup \Sigma, f' ) \}, \\
&D'_0 = d,
\end{align*}
and
\begin{eqnarray*}
\left\{ \begin{array}{ll}
g(a) = a' \qquad (\text{for }a \in \Delta),  \\
g(a) = a \qquad \, (\text{for }a \in S' - \Delta).
\end{array} \right.
\end{eqnarray*}
Let $w  = b_{(a_1, 1)} \cdots b_{(a_1, |h(a_1)|)} \cdots b_{(a_n, 1)} \cdots b_{(a_n, |h(a_n)|)} \in L_{\lambda}(\mathcal{A})$. Hence, $a_1 \cdots a_n$ is included in $h^{-1}(w)$.
Moreover, let $D_{i} \rightarrow^{b_{(a_m, 1)}} \cdots \rightarrow^{b_{(a_m, |h(a_m)|)}} D_{j}$ be a part of $\pi \in IP_{\lambda}(\mathcal{A}, w)$. For $|h(a_m)| \ge 3$, it is imitated in $\pi' \in IP_{\lambda}(\mathcal{A}', w)$ as follows:
\begin{align*}
&g(D_{i-1}) \\
\rightarrow^{a_m} &g(D_i) b'_{(a_m, 1)} q_{(a_m, 1)} \\  
\rightarrow^{\lambda} \ \, &g(D_{i+1}) b'_{(a_m, 2)} q_{(a_m, 2)} \rightarrow^{\lambda} \cdots \\
&(\text{or } \rightarrow^{\lambda} \ \, g(D_{i+1}) q_{(a_m, 1)} \rightarrow^{\lambda} \cdots, \text{ for a $\lambda$-input in } \mathcal{A} ) \\
\rightarrow^{\lambda} \ \, &g(D_{i-1}) b'_{(a_m, h(a_m))} \\
\rightarrow^{\lambda} \ \, &g(D_{i}).
\end{align*}

The other direction of the proof is shown in the similar manner. Hence, it holds that $h^{-1}( L_{\lambda}(\mathcal{A}))  = L_{\lambda}(\mathcal{A'})$ and the workspace of $\mathcal{A'}$ is linear-bounded.

[$\lambda$-free morprhisms]
We first show that $\lambda$-$\mathcal{LRA}$ is closed under codings.
For an LRA $\mathcal{A} = (S, \Sigma, A, D_0, f)$ in $\lambda$-normal form and a coding $h:\Sigma^* \rightarrow \Delta^*$, construct an RA $\mathcal{A'} = (S', \Delta, A', D'_0, f')$ and a mapping $h: S'^{\#} \rightarrow S'^{\#}$  as follows:
\begin{align*}
&S' = S \cup \Sigma' \cup \Delta \cup \{ d \}, \text{ where } \Sigma' = \{ a' \, | \, a \in \Sigma \}, \\
&A' = \{ ( h(R), h(I), h(P) ) \, | \, ( R, I, P ) \in A \} \\
&\qquad \cup \{ ( h(a), \emptyset, a' ) \, | \, a \in \Sigma \}  \\
&\qquad \cup \{ ( d, \emptyset, h(D_0) ) \},  \\
&D'_0 = d,
\end{align*}
and
\begin{eqnarray*}
\left\{ \begin{array}{ll}
h(a) = a' \qquad (\text{for }a \in \Sigma),  \\
h(a) = a \qquad \, (\text{for }a \in S' - \Sigma).
\end{array} \right.
\end{eqnarray*}
Then, it holds that $h(L_{\lambda}(\mathcal{A})) = L_{\lambda}(\mathcal{A'})$ and the workspace of $\mathcal{A'}$ is linear-bounded.

In Theorem 3.7.1 of \cite{SG:75}, it is shown that each family closed under inverse morphisms, intersection with regular languages and codings is also closed under $\lambda$-free morprhisms. Hence, $\lambda$-$\mathcal{LRA}$ is closed under $\lambda$-free morprhisms.

[$\lambda$-free gsm-mappings]
Since every trio is closed under $\lambda$-free gsm-mappings (\cite{AS:73}), $\lambda$-$\mathcal{LRA}$ is closed under $\lambda$-free gsm-mappings.
\end{proof}

We shall show that $\lambda$-$\mathcal{LRA}$ shares common negative closure properties with $\mathcal{LRA}$. The manner of proving those results is almost parallel to that of proofs for $\mathcal{LRA}$ presented in the previous section. In order to make this paper self-contained,  below we give the proof of the following lemma that is a $\lambda$-version of 
Lemma 2 (i.e., of Lemma 1 in \cite{OKY:12}). 
 
\begin{lem}
For an alphabet $\Sigma$ with $|\Sigma| \ge 2$, let $h:\Sigma^* \rightarrow \Sigma^*$ be an injection such that for any $w \in \Sigma^*$, $|h(w)|$ is bounded by a polynomial of $|w|$. Then, there is no PRA $\mathcal{A}$ such that $L_{\lambda}(\mathcal{A}) = \{ wh(w) \, | \,  w \in \Sigma^* \}$. 
\label{lem-ww-lambda}
\end{lem}

\begin{proof}
Assume that there is a {\it poly}-RA $\mathcal{A} = (S, \Sigma, A, D_0, f)$ such that $L_{\lambda}(\mathcal{A}) = \{ wh(w) \, | \,  w \in \Sigma^* \}$.  Let $|S| = m_1$, $|\Sigma|= m_2 \ge 2$ and the input string be $wh(w)$ with $|w|=n$. 

Since $|h(w)|$ is bounded by a polynomial of $|w|$, $|wh(w)|$ is also bounded by a polynomial of $n$.  Hence, for each $D_i$ in an interactive process $\pi \in IP_{\lambda}( \mathcal{A}, wh(w) )$, it holds that $|D_i| \le p(n)$ for some polynomial $p(n)$ from the definition of a {\it poly}-RA.

Let $\mathcal{D}_{p(n)} = \{ D \in S^\# \, | \, |D| \le p(n) \}$. Then, it holds that
\begin{align*}
&|\mathcal{D}_{p(n)}| = \sum^{p(n)}_{k = 0} {}_{m_1} \mathrm{H}_k = \sum^{p(n)}_{k = 0} \frac{(k + m_1 -1)!}{k! \cdot (m_1 -1) !} = \frac{(p(n) + m_1)!}{p(n)! \cdot m_1 !}  \\
&\hspace{11mm} = \frac{(p(n) + m_1)(p(n) + m_1 - 1) \cdots (p(n) + 1)}{ m_1 !} \tag{$*$}
\end{align*}
where ${}_{m_1} \mathrm{H}_k$ denotes the number of repeated combinations of $m_1$ things taken $k$ at a time. Therefore, there is a polynomial $p'(n)$ such that $|\mathcal{D}_{p(n)}| = p'(n)$. Since it holds that $|\Sigma^n| = (m_2)^n$, if $n$ is sufficiently large,  we obtain the inequality $|\mathcal{D}_{p(n)}| < |\Sigma^n|$. 

For $w = a_1 \cdots a_n \in \Sigma^*$, let $I(w) = \{ D \in \mathcal{D}_{p(n)} \, | \, \pi = D_0 \rightarrow^{a_1} \cdots \rightarrow^{a_n} D \rightarrow \cdots \in IP_{\lambda}(\mathcal{A}, w) \} \subseteq \mathcal{D}_{p(n)}$, i.e., $I(w)$ is the set of multisets in $\mathcal{D}_{p(n)}$ which appear immediately after inputing $w$ in $IP_{\lambda}(\mathcal{A}, w)$. From the fact that $L(\mathcal{A}) = \{ wh(w) \, | \,  w \in \Sigma^* \}$ and $h$ is an injection, we can show that for any two distinct strings $w_1, w_2 \in \Sigma^n$, $I(w_1)$ and $I(w_2)$ are incomparable. This is because if $I(w_1) \subseteq I(w_2)$, then the string $w_2 h(w_1)$ is in $L_{\lambda}(\mathcal{A})$, which means that $h(w_1) = h(w_2)$ and contradicts that $h$ is an injection. 

Since for any two distinct strings $w_1, w_2 \in \Sigma^n$, $I(w_1)$ and $I(w_2)$ are incomparable and $I(w_1), I(w_2) \subseteq \mathcal{D}_{p(n)}$, it holds that 
\[ | \{ I(w) \, | \, w \in \Sigma^n \} | \le |\mathcal{D}_{p(n)}| < |\Sigma^n|. \]
However, from the pigeonhole principle, the inequality $| \{ I(w) \, | \, w \in \Sigma^n \} | < |\Sigma^n|$ contradicts that for any two distinct strings $w_1, w_2 \in \Sigma^n$, $I(w_1) \ne I(w_2)$. Hence, there is no LRA $\mathcal{A}$ such that $L_{\lambda}(\mathcal{A}) = \{ wh(w) \, | \,  w \in \Sigma^* \}$.
\end{proof}

\begin{thm}
$\lambda$-$\mathcal{LRA}$ is not closed under complementation, quotient by regular languages, morphisms or gsm-mappings.
\end{thm}

\begin{cor}
$\lambda$-$\mathcal{LRA}$ is an AFL, but not a full AFL.
\end{cor}

\begin{flushleft}
Remark: We note the class $\lambda$-$\mathcal{PRA}$ could be proved to be an AFL in the same manner as $\lambda$-$\mathcal{LRA}$.
\end{flushleft}

\section{Further characterizations of $\mathcal{LRA}$ and $\mathcal{ERA}$}

In this section, we develop further characterizations concerning $\mathcal{LRA}$ and $\mathcal{ERA}$ in relation to the Chomsky hierarchy, and show two interesting results. One is concerned with a representation theorem for the class $\mathcal{RE}$ in terms of $\mathcal{LRA}$, and the other is a new characterization of $\mathcal{CS}$ with $\mathcal{ERA}$.

\begin{thm}
For any context-sensitive language $L \subseteq \Sigma^*$, there exists an LRA $\mathcal{A}$ such that $w \in L$ if and only if $c^{2^n} w \in L(\mathcal{A})$ {\rm (}or $w c^{2^n} \in L(\mathcal{A})${\rm )} with $|w|=n$ and $c \notin \Sigma$.
\end{thm}

\begin{proof}
From Proposition \ref{prop-ra}, let $\mathcal{A} = (S, \Sigma, A, D_0, f)$ be an ERA which accepts $L$. Then, construct an RA $\mathcal{A'} = (S', \Sigma \cup c, A', D'_0, f')$ as follows:
\begin{align*}
&S' = S \cup \{ p_0, p_1, p_2, p_3, p_4, n_1, n_2, c, c_1, c_2, d, e, f', f'' \}, \\
&A' = \{ R, I \cup \{ c, f' \}, P ) \, | \, ( R, I, P ) \in A \} \\
&\qquad \cup \{ {\bf a}'_1, {\bf a}'_2, {\bf a}_3, {\bf a}_4, {\bf a}_5, {\bf a}_6, {\bf a}_7, {\bf a}'_8, {\bf a}_9, {\bf a}_{10}, {\bf a}_{11}, {\bf a}_{12} \}, \ \mbox{where} \\
&\quad  {\bf a}'_1 = ( p_0, c, p_1 p_2 p_3 n_2 ),\ {\bf a}'_2 = ( p_1, eff'f'', p_1 n_1 ),\ {\bf a}_3 = ( c, p_1, c_1 ),\ {\bf a}_4 = ( {c^2_1}, p_0 c_2 e, c_2 ),  \\
&\quad {\bf a}_5 = ( {c^2_2}, p_0 c_1 e, c_1 ), \ {\bf a}_6 = ( c_1 d, p_0 c_2, e ),\ {\bf a}_{7} = ( c_2 d, p_0 c_1, e ),\ {\bf a}'_{8} = ( e, p_0 c c_1 c_2, f'' ), \\
&\quad {\bf a}_9 = ( p_2, cp_4, p_2 n_2 ), \ {\bf a}_{10} = ( p_3, \Sigma, p_4 ), \ {\bf a}_{11} = ( n_1 n_2, \emptyset, \lambda ), \ {\bf a}_{12} = ( ff'', p_3 n_1 n_2, f' ), \\
&D'_0 = D_0 + dp_0.
\end{align*}

Note that the reactions ${\bf a}'_1$-${\bf a}'_8$ are almost the same as the ones of Example \ref{exam-c2n}. Therefore, the total number of $n_1$ appearing in a interactive process of $IP(\mathcal{A'}, c^{2^n}w)$ is $n+1$ (see Example \ref{exam-c2n} and Figure \ref{fig-c2n}). 
On the other hand, the total number of $n_2$ appearing in a interactive process of $IP(\mathcal{A'}, c^{2^n}w)$ is $|w|+1$, which is derived by the reactions ${\bf a}'_1, {\bf a}_9, {\bf a}_{10}$. Using the reaction ${\bf a}_{11}$, it is confirmed that if $c^{2^n}w$ is accepted by $\mathcal{A'}$, then $n+1 = |w|+1$.
Hence, it holds that $w \in L(\mathcal{A})$ if and only if $c^{2^n} w \in L(\mathcal{A'})$ with $|w|=n$.  

Since the workspace of $\mathcal{A}$ for $w$ is bounded by an exponential function with respect to the length $|w|=n$, the workspace of $\mathcal{A'}$ for $c^{2^n} w$ is bounded by a linear function with respect to the length $|c^{2^n} w| = 2^n + n$.  

For the case $wc^{2^n}$, we can prove in a similar manner.
\end{proof}

\begin{thm}[Theorem 3.12 in \cite{PRS:98}]
For any recursively enumerable language $L \subseteq \Sigma^*$, context-sensitive language $L'$ such that $w \in L$ if and only if ${c_2}^i c_1 w \in L'$ (or $w c_1 {c_2}^i \in L'$) for some $i \ge 0$ and $c_1, c_2 \notin \Sigma$. 
\end{thm}

\begin{cor}
For any recursively enumerable language $L \subseteq \Sigma^*$, there exists an LRA $\mathcal{A}$ such that $w \in L$ if and only if ${c_3}^j {c_2}^i c_1 w \in L(\mathcal{A})$ (or $w c_1 {c_2}^i {c_3}^j \in L(\mathcal{A})$) for some $i,j \ge 0$ and $c_1, c_2, c_3 \notin \Sigma$. 
\end{cor}

\begin{cor}
(i) For any recursively enumerable language $L$, there exists an LRA $\mathcal{A}$ such that $L = R \verb|\|  L(\mathcal{A})$ (or $L(\mathcal{A}) / R$) for some regular language $R$. \\
(ii) For any recursively enumerable language $L$, there exists an LRA $\mathcal{A}$ such that $L = h( L(\mathcal{A}))$ for some projection $h$.
\label{cor-re-lra}
\end{cor}

\begin{thm}
For a language $L$, $L$ is a context-sensitive language if and only if $L$ is accepted by an ERA.
\end{thm}

\begin{proof}
The claim of ``only if'' part holds by Proposition \ref{prop-ra}.

Let $S = \{ s_1, \ldots, s_k \}$ be an ordered alphabet and $\mathcal{A} = (S, \Sigma, A, D_0, f)$ be an ERA. Assume that for an input $w = a_1 \cdots a_n \in \Sigma^*$, the workspace of $\mathcal{A}$ is bounded by the exponential function $s(n) = c_1 c^n_2$, where $c_1, c_2 \ge 0$ are constants. Then, we shall construct the nondeterministic $(k+2)$-tape linear-bounded automaton $M_{\mathcal{A}}$ in which the length of each tape is bounded by $ckn$ for some constant $c$.  $M_{\mathcal{A}}$ imitates an interactive process $\pi: D_0, \ldots, D_n, \ldots \in IP(\mathcal{A}, w)$ in the following manner:  
\begin{enumerate}
\item At first, Tape-$1$ has the input $w \in \Sigma^*$ and Tape-$(i+1)$ has the number of $s_i$ in $D_0$ (for $1 \le i \le k$)  represented by $c_2$-ary number. Tape-$(k+2)$ is used to count the number of computation step of $M_{\mathcal{A}}$.
\item Let $D$ be the current multiset in $\pi$. When $M_{\mathcal{A}}$ reads the symbol $s_i$ in the input, add one to the Tape-$(i+1)$. Then, by checking all tapes except Tape-$1$, compute an element of $Res_A(s_i +D)$ in the nondeterministic way and rewrite the contents in the tapes. After reading through the input $w$, $M_{\mathcal{A}}$ computes an element of $Res_A(D)$ in the nondeterministic way and rewrite the contents in the tapes.
\item After reading through the input $w$, if $Res_A(D) = \{ D \}$ and $f \subseteq D$, then $M_{\mathcal{A}}$ accepts $w$. In the case where (i) $Res_A(D) = \{ D \}$ and $f \nsubseteq D$, (ii) $|D|$ exceeds $c_1 c^n_2$ or (iii) the number of computation step exceeds $c_3 c^{kn}_2$ for $k(=|S|)$ and some constant $c_3$, $M_{\mathcal{A}}$ rejects $w$.
\end{enumerate}
Since we use $c_2$-ary number for counting the number of symbols, the length  $\log_{c_2}(c_1) + n$ of each tape is enough to memorize $D$ with $|D| \le c_1 c^n_2$. 
In the case where $M_{\mathcal{A}}$ never stops with the input $w$, there exists a cycle of configurations in the computation of $M_{\mathcal{A}}$. Since the number of all possible $D$s during the computation is bounded by $c_3 c^{kn}_2$ for $k$ and some constant $c_3$ (see the equation ($*$) in the proof of Lemma \ref{lem-ww-lambda}),  the length of Tape-$(k+2)$ to count the number of steps of computation is bounded by $\log_{c_2}(c_3) + kn$.
Therefore, it holds that $ L(M_\mathcal{A}) = L(\mathcal{A})$.
\end{proof}

Table 1 summarizes the results of closure properties of both 
$\mathcal{LRA}$ and $\lambda$-$\mathcal{LRA}$, while Figure \ref{hierarchy} 
illustrates the relationship between language classes defined by 
a various types of bounded reaction automata and the Chomsky hierarchy.

\begin{table}
\caption{Closure properties of $\mathcal{LRA}$ and $\lambda$-$\mathcal{LRA}$.}
\begin{center}
\begin{tabular}{|l|c|c|} \hline
language operations & \ $\mathcal{LRA}$ \ & \ $\lambda$-$\mathcal{LRA}$ \  \\ \hline \hline
union & Y & Y \\ \hline 
intersection & Y & Y \\ \hline
complementation & N & N \\ \hline
concatenation & Y & Y \\ \hline
Kleene $+$ & ? & Y \\ \hline
Kleene $*$ & ? & Y \\ \hline
(right \& left) derivative & Y & Y \\ \hline
(right \& left) quotient by regular languages & N & N \\ \hline
$\lambda$-free morphisms & Y & Y \\ \hline
morphisms & N & N \\ \hline
inverse morphisms & ? & Y \\ \hline
$\lambda$-free gsm-mappings & Y & Y \\ \hline
gsm-mappings & N & N \\ \hline
shuffle & Y & Y \\ \hline
\end{tabular}
\end{center}
\end{table}

\begin{figure}[t]
\centerline{
\includegraphics[scale=0.55]{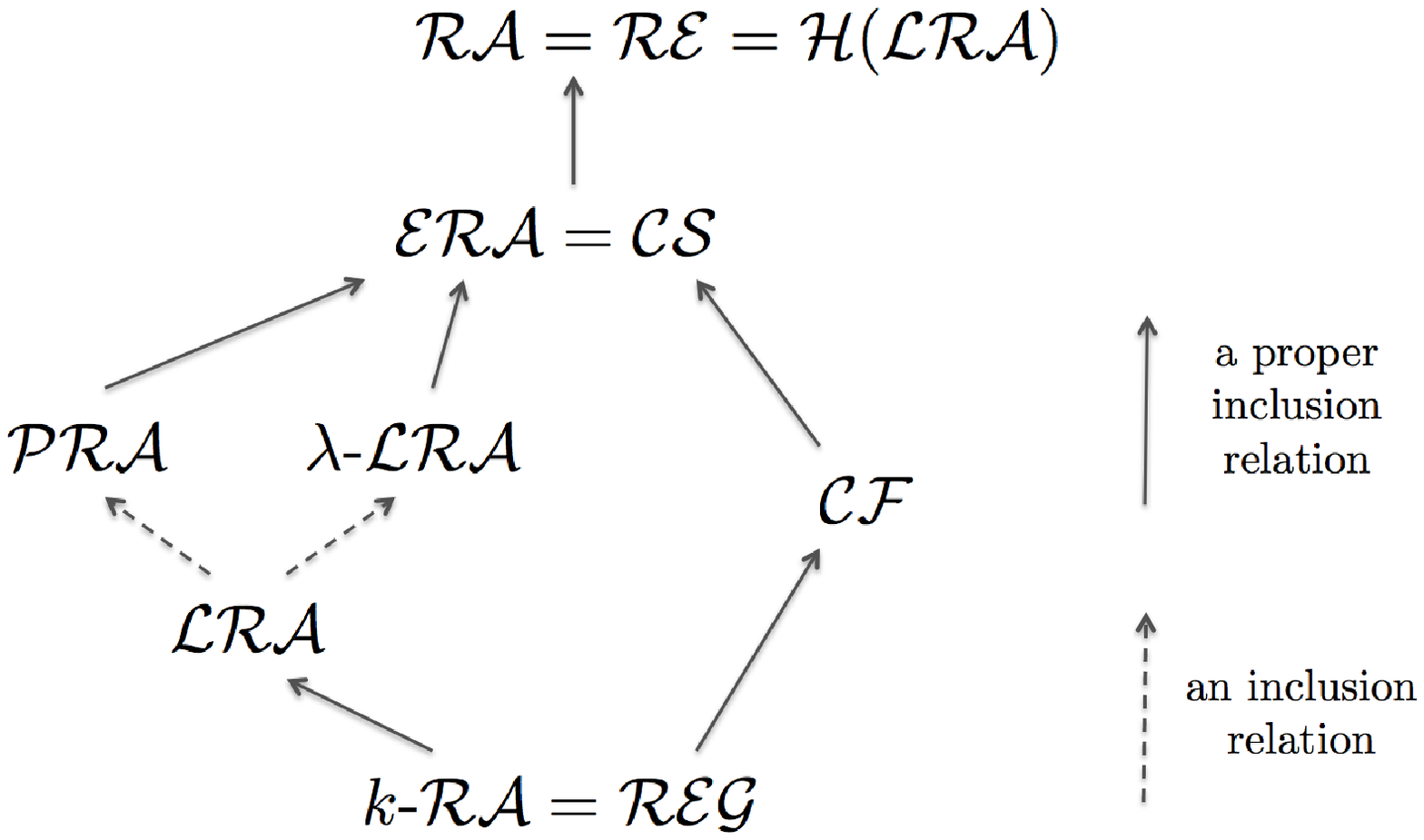}}
\caption{The hierarchy of the language classes accepted by bounded reaction automata, where $\mathcal{H(LRA)}= \{ h(L) \, | \, L \in \mathcal{LRA}, h \text{ is a projection} \}$.}
\label{hierarchy}
\end{figure}

\section{Conclusion}

We have continued the investigation of reaction automata (RAs)  with  a focus on the formal language theoretic properties of subclasses of RAs, called  linear-bounded RAs (LRAs) and  exponentially-bounded RAs (ERAs).   Besides LRAs, we have newly introduced an extended model 
(denoted by $\lambda$-LRAs) by allowing $\lambda$-moves in the accepting process of reaction, and investigated the closure properties of language classes $\mathcal{LRA}$ and $\lambda$-$\mathcal{LRA}$. 
We have shown the following :  \\
\quad (\,i\,) the class  $\lambda$-$\mathcal{LRA}$ forms an AFL with additional closure properties, \\
\quad (\rnum{2})  any recursively enumerable language can be expressed as  
a homomorphic image of a language in $\mathcal{LRA}$, \\
\quad (\rnum{3}) the class  $\mathcal{ERA}$  coincides with the class of context-sensitive languages. 

Considering  the definitions of  the existing acceptors in the classical automata theory, the result (\,i\,) suggests that it may be  reasonably justifiable to allow each  subclass of bounded RAs to have $\lambda$-transitions in its definition. 
 From the result (\rnum{2}) and the closure properties (shown in Table 1), it is interesting to see that the class  $\mathcal{LRA}$  (or $\lambda$-$\mathcal{LRA}$) is closer to the class $\mathcal{CS}$ rather than the class $\mathcal{CF}$ in its language theoretic property.  Further,  
an intriguing  result (\rnum{3}) together with the result that $\mathcal{RE}=\mathcal{RA}$ (in \cite{OKY:12}) may provide a new insight into the theory of computational complexity.  

Many subjects remain to be investigated. First, it is open whether or not $\mathcal{LRA}$ is equal to $\lambda$-$\mathcal{LRA}$, whose positive answer immediately settles open problems of the closure properties on  
$\mathcal{LRA}$ (see Table 1).  Also, we do not know the proper inclusion relation between $\mathcal{LRA}$ and $\mathcal{PRA}$. Secondly, it is interesting to explore the relationship between RAs and other computing devices that are based on the multiset rewriting, such as a variety of P-systems and their variants (\cite{PP:11}),  Petri net models (\cite{HM:01}).  Also, it would be useful to develop a method for simulating a variety of chemical reactions in the real world by the use of the framework based on reaction automata.

\end{document}